\documentclass[12pt]{article}
\usepackage{setspace}
\usepackage[margin=1.25in]{geometry}

\usepackage[utf8]{inputenc}
\usepackage{hanging}	
\usepackage{soul}
\usepackage{xcolor}

\usepackage{amsmath}
\usepackage{amsthm}
\usepackage{amssymb}
\usepackage{amsfonts}
\usepackage{setspace}
\usepackage{multirow}
\usepackage{caption}
\DeclareCaptionLabelFormat{AppendixTables}{A.#2}
\usepackage{subcaption}
\usepackage{graphics}
\usepackage{lscape}
\usepackage{changepage}
\usepackage{graphicx}
\usepackage{xcolor,soul}
\usepackage{booktabs}
\usepackage{hyperref}
\usepackage{tikz-cd}
\usepackage[makeroom]{cancel}
\usepackage{mathrsfs}
\usepackage{natbib}
\usepackage{xr-hyper}
\usepackage{subfiles}
\tikzset{
  symbol/.style={
    draw=none,
    every to/.append style={
      edge node={node [sloped, allow upside down, auto=false]{$#1$}}}
  }
}

\theoremstyle{definition}
\newtheorem{theorem}{Theorem}[section]  \newtheorem{assumption}{Assumption}[section] \newtheorem{remark}{Remark}[section] \newtheorem{lemma}{Lemma}[section] \newtheorem{proposition}{Proposition}[section] \newtheorem{definition}{Definition}[section]
\newtheorem{corollary}{Corollary}[section]

\title{Informativeness under Model Uncertainty: Shadow Prices and Ridge Penalties}
\author{Jieun Lee\footnote{Department of Economics, Emory University. \textit{jieun.lee@emory.edu; jieunlee.sophia@gmail.com}}\; and Esfandiar Maasoumi\footnote{Department of Economics, Emory University. \textit{esfandiar.maasoumi@emory.edu}}}

\begin{document}
\onehalfspacing

\maketitle

\begin{abstract} 
We develop inference under model uncertainty due to weak, noisy, multiple candidate restrictions and theories, and nuisance control covariates. A unified framework is given with degrees of misspecification and corresponding shadow prices, based on a Lagrangian constrained optimization approach, and a data-driven tolerance parameter selected via a Stein-type (shrinkage) risk criterion. A debiasing step is based on Karush--Kuhn--Tucker conditions. We introduce individual shadow prices (ISP) for different restrictions to measure empirical relevance and propose a plateau rule to separate signal from noise. We establish consistency and asymptotic normality of the estimators and characterize the ISP. Simulations and an application to a Solow growth model illustrate the method’s practical usefulness.

\noindent Keywords: model uncertainty; multiple restrictions; individual shadow prices; regularization bias; Stein-type risk; debiasing; structural break detection. \\
\noindent JEL codes: C12, C13, C52, C61.
\end{abstract}

\section{Introduction}
Empirical research often involves competing theories under multiple model restrictions. Importance of accommodating this \textit{model uncertainty} is gaining overdue attention. We focus on a setting with potentially invalid or approximate restrictions on a target component of models, $m_1$, say, with possibly further nuisance components. Our goal is to assess degrees of compatibility with the data evidence, priced by shadow prices, avoiding suboptimal model selection strategies on $m_1$.

This perspective aligns with semiparametric and modern machine learning frameworks under sparsity. We view the model as consisting of a target component $m_1$, a nuisance component $m_2$, managed by Machine Learning, and an error term (see, e.g., \citet{Chernozhukovetal2018}). After accounting for the nuisance component $m_2$, inference reduces to the finite-dimensional target $m_1$, on which economically meaningful, theory-driven restrictions are imposed. This is a generalization of partially linear models and General Linear Models (GLM). Methods for robust and debiased ML, prior to inference on $m_1$ abound. See \citet{DrukkerandLiu2022} for a recent review and algorithms, highlighting the centrality of Neyman Orthogonalization and computational time and efficiency.

Empirical decisions are sensitive to weak or misspecified restrictions. This is the case whether restrictions produce greater sparsity or enlarge predictive sets by irrelavant variables; see Gospodinov, Kan, and Robbati (\citeyear{Gospodinovetal2014})-GKR. Fundamental tradeoff, imposing restrictions may introduce bias, while discarding them sacrifices efficiency, is exacerbated by possibility of limiting properties of statistics, see GKR (\citeyear{Gospodinovetal2014}) and \citet{MaasoumiandPhillips1982}. This tension mirrors the bias--variance tradeoff inherent in modern regularization and shrinkage methods.

We propose a unified framework for estimation and the assessment of informativeness (empirical relevance) under multiple restrictions and model uncertainty. Rather than selecting a single model, our approach allows for controlled degrees of misspecification and lets the data determine the empirical importance of each restriction. Specifically, the framework delivers consistent estimation in the presence of multiple candidate restrictions while simultaneously quantifying the empirical relevance of each restriction for model fit. In this sense, we move from model selection toward a structured form of model averaging or mixed estimation (of $m_1$) that better reflects practical empirical work.

Initial assessment of restrictions based on their empirical relevance precedes traditional testing. This effectively prioritizes informative restrictions and downweights or discards those that are weak or severely misspecified. The approach is consistent with the spirit of modern double machine learning, where nuisance components are flexibly controlled and inference focuses on a low-dimensional target. By incorporating a data-driven assessment of restrictions, the framework reduces the multiplicity burden while maintaining a fixed error rate, thereby improving the power of subsequent inference.

At the estimation stage, we adopt a bi-level (profiling) algorithm. The inner problem solves a Lagrangian constrained optimization for a given tolerance level governing restriction misspecification, yielding the constrained estimator and its associated shadow price. The outer problem selects the tolerance parameter by optimizing a Stein-type risk criterion, allowing the estimator to adapt to the empirical environment and optimally balance regularization bias and variance. Importantly, this approach does not require prior knowledge of which restrictions are valid or informative and is robust to weak or misspecified restrictions.

We further refine estimation by correcting a residual centering bias. Although the optimized tolerance parameter achieves an optimal bias--variance tradeoff, it does not eliminate the regularization bias induced by the constraint. We therefore introduce a debiasing step based on a local characterization of this bias. From the Karush--Kuhn--Tucker conditions, the deviation of the pseudo-true parameter from the unconstrained target is proportional to the Lagrange multiplier, which captures the local tightness of the constraint. This leads to a simple correction that removes the leading distortion and recenters the estimator without affecting its first-order variance.

Building on this estimation framework, we assess the empirical relevance of restrictions through \emph{individual shadow prices} (ISP), which decompose the global shadow price into restriction-specific components. These ISP measure the marginal gain from relaxing each restriction and provide a coherent and comparable metric of empirical relevance. We then propose a \emph{plateau rule} to separate signal from noise by detecting a structural break in the ordered ISP sequence.

We establish consistency and asymptotic normality of the constrained estimator with optimized misspecification tolerance and its debiased counterpart, and characterize the asymptotic behavior of the ISP. Notably, the constrained optimization approach yields an estimator that is robust to restriction misspecification, as the tolerance parameter allows the data to determine the extent to which restrictions are enforced. The asymptotic analysis formalizes this property: the optimized tolerance converges to a pseudo-true value governing the bias--variance tradeoff, so that the constrained estimator targets a pseudo-true parameter reflecting both sampling variation and misspecification. While this does not eliminate regularization bias, the debiasing step corrects the leading distortion and restores convergence to the structural target.

Simulation results demonstrate strong finite-sample performance: the proposed estimator is consistent, analytical inference attains near-nominal coverage comparable to the bootstrap, and the plateau rule effectively separates signal from noise. The empirical application highlights a key insight: statistical validity and empirical relevance are distinct. A restriction may be rejected by standard tests yet remain empirically irrelevant, as relaxing it does not improve model fit. This highlights the value of the proposed approach in identifying which restrictions matter for estimation.

The rest of the paper is organized as follows. In Section 1.1, we discuss the relation of our framework to the existing literature. In Section 2, we present our theory framework. In Section 3, we develop the asymptotic theory. In Section 4, we develop bootstrap-based inference as an extension. In Section 5, we report Monte Carlo simulation results. In Section 6, we provide an empirical application. In Section 7, we conclude. Additional discussions are provided in the Supplement.

\subsection{Related literature}
\paragraph{Restrictions from economic theory.} We consider a priori restrictions on the parameter space and through statistical distributions (e.g., moment restrictions). Direct restrictions may be imposed in restricted estimation, and implied restrictions may be tested by diagnostic testing. 

For example, in macroeconomics and growth, structural approaches such as \citet{Solow1956} and \citet{Swan1956}, the human-capital-augmented growth regressions of \citet{MankiwRomerandWeil1992}, and the spatial growth model of \citet{ErturandKoch2007} incorporate cross-parameter restrictions implied by theory directly into the regression specification. Similarly, in microeconomics and industrial organization, rational expectations models \citep{HansenandSargent1980} and auction models based on Bayesian Nash equilibrium \citep{GuerrePerrigneandVuong2000} embed equilibrium restrictions within the estimating equations.

Other ``restrictions'' are implied and treated as \emph{testable implications of an empirical model}. In this case, estimation proceeds without imposing the full set of theoretical restrictions, and the restrictions are evaluated ex post using the data. Examples include present-value tests based on vector autoregressions \citep{CampbellandShiller1987}, revealed-preference and shape-restriction tests in demand analysis \citep{Blundell2005}, and asset pricing tests that assess zero-alpha conditions or bounds on risk aversion \citep{FamaandMacBeth1973,MehraandPrescott1985}.

We provide a unified, data-driven method to rank and select restrictions prior to formal validity testing. The proposed constrained optimization framework accommodates a broad class of linear and nonlinear restrictions.

\paragraph{Model uncertainty, ambiguity, and misspecification.}
Model uncertainty arises when multiple candidate restrictions are considered as if they were (approximately) valid. Constrained models may be within a limited ball of deviating candidate models (model ambiguity), or globally misspecified. Metrics and different rates for deviations are necessary. It has been increasingly recognized that shrinkage methods, with or without pretesting interpretation, provide an actionable approach to formalizing degrees of misspecification.  

Early work by \citet{Maasoumi1978} proposed combining estimators under model uncertainty using information from test statistics. More recent literature on model ambiguity and sensitivity (e.g., \citet{BonhommeandWeider2022}, \citet{ChristensenandConnault2023}) treats such uncertainty as intrinsic, allowing models or distributions to vary within neighborhoods of a benchmark specification. These approaches quantify the sensitivity of conclusions to misspecification, often through a notion of \emph{misspecification size}. For instance, \citet{ChristensenandConnault2023} derive counterfactual bounds by optimizing over divergence-based neighborhoods. Our framework builds on this perspective but differs in a key dimension. Rather than varying distributions or moment conditions directly, we parameterize misspecification through a tolerance level $c$ on restriction violations and select it endogenously. This shifts the focus from evaluating robustness over a fixed neighborhood to optimizing the degree of deviation from imposed restrictions.

A context is provided by the Hansen--Jagannathan distance (GKR, \citeyear{Gospodinovetal2013,Gospodinovetal2016}), which provides a geometric measure of misspecification. While the HJ-distance treats misspecification as an \emph{ex post} phenomenon, our framework treats it as a \emph{decision variable}. Moreover, Lagrange multipliers are central in our analysis and are interpreted as ISP, which quantify the marginal contribution of each restriction. A bound on a constraining hyper sphere (or elipsoid) which represents tolerance for model deviations (see below for definition of ``c''), is optimized in the spirit of Stein-Like estimators, for an explicit bias--variance tradeoff, extending existing robustness frameworks to an optimization-based approach to model uncertainty and misspecification.

In contrast to local misspecification as the key feature of model ambiguity, global or other considerations of misspecification raise the question of model selection and model averaging. In the ML context, LASSO is an example of model selection, whereas Ridge is a form of model averaging. Model selection entails biases and generally requires presumption of a reference as ``the true DGP''. This is known to be suboptimal.\footnote{Bayesian methods are also generally model selection procedures, but may be extended to model averaging.} Model averaging can avoid these shortcomings at the cost of adopting an oracle model average object; see \citet{GospodinovandMaasoumi2021}.  

\paragraph{Inference-based insights for regularization methods in machine learning.}
Our framework is closely related to machine learning literature through the duality between constrained optimization and penalized estimation \citep{HoerlandKennard1970,Tibshirani1996,ZouandHastie2005}. This duality implies that the estimator can be viewed as solving a constrained problem in which deviations from imposed restrictions are controlled within a tolerance level, thereby allowing for restriction misspecification and enhancing robustness. Importantly, while the machine learning literature typically addresses regularization bias through tuning parameters, our approach treats it as a statistical optimization problem, using a Stein-type risk criterion for selecting the tolerance parameter and a debiasing correction based on analytic theory. Optimization techniques for choice of ``tuning parameters'', such as by Bayesian Information Criterion (BIC), compete with Cross Validation (CV) and other ad hoc methods.

Our framework extends these methods to an inferential setting and contributes to the growing literature on inference in machine learning. While inference tools for machine learning models, such as post hoc decomposition methods like SHAP \citep{LundbergandLee2017}, quantify the contribution of individual covariates to predictions from a fixed model, our approach develops statistical decision rules for assessing the relative empirical support for restrictions (that may be false). This is consistent with our preference to avoid assumption of a ``true DGP''.

\section{Theory framework} \label{sec:model_spec}
Let \(g(\theta) \in \mathbb{R}^q\) denote a vector function of the parameter \(\theta \in \mathbb{R}^p\), and let \(\Sigma \in \mathbb{R}^{q \times q}\) be a known positive definite weighting matrix. We represent the restriction system through the quadratic form
$
h(\theta)=g(\theta)^{\prime}\Sigma^{-1}g(\theta).$

We introduce a finite upper bound $c_0 > 0$ on $h(\theta)$, from a candidate set \(\mathcal C \subset (0,c_0]\) of tolerance levels. For each \(c \in \mathcal C\), the scalar \(c\) determines how tightly the restrictions are enforced: small values impose strong adherence, while larger values allow greater deviation and thus provide robustness to misspecification. Rather than fixing \(c\) \textit{a priori}, we select it from the data using a Stein-type analytical risk criterion.
These restrictions represent many test statistics and penalization functions, including the canonical Ridge. The scalar $c$ resembles the critical level of frequentist testing procedures, but has counterparts in Bayesian prior interpretations. 

Our framework is cast in a general M-estimation setting. Let
$
\phi_n(\theta)=\frac{1}{n}\sum_{i=1}^n \phi(Z_i,\theta)$,
$\phi(\theta)=\mathbb{E}[\phi(Z_{i},\theta)]$, and
$\phi^0=\mathbb{E}[\phi(Z_i,\theta^{0})],
$
where \(Z_i=(y_i,x_i)\) are i.i.d.\ observations, and \(\phi(\cdot)\) is a loss function. For each \(c \in \mathcal C\), define the constrained estimator
\begin{equation}
\hat{\theta}(c)
\in
\arg\min_{\theta \in \Theta} \phi_n(\theta)
\quad
\text{s.t.}
\quad
h(\theta)\le c.
\label{eq:general_problem}
\end{equation}

\noindent The associated Lagrangian is
\begin{equation}
\mathcal{L}_n(\theta,\lambda;c)
=
\phi_n(\theta)
+
\lambda\bigl(h(\theta)-c\bigr),
\qquad
\lambda\ge 0,
\end{equation}
with first-order condition
$
\nabla_\theta \phi_n(\theta)
+
2\lambda \nabla_\theta g(\theta)^{\prime}\Sigma^{-1}g(\theta)
=
0.
$ Thus, for each \(c\), the pair \((\hat{\theta}(c),\hat{\lambda}(c))\) is determined jointly from the Karush--Kuhn--Tucker (KKT) system. The multiplier \(\hat{\lambda}(c)\) measures the marginal gain from relaxing the restriction and admits a natural interpretation as a shadow price.

The outer step selects the tolerance level via
\begin{equation}
\hat c \in \arg\min_{c \in \mathcal C} \widehat R_{\mathrm{lin}}(c),
\end{equation}
yielding the optimized estimator
$
\hat{\theta}=\hat{\theta}(\hat c),\;
\hat{\lambda}=\hat{\lambda}(\hat c).
$ The procedure is therefore bi-level: the inner problem computes \(\hat{\theta}(c)\), while the outer problem selects \(c\) to balance regularization bias and variance.

\paragraph{Population characterization.}
Let \((\theta^*(c),\lambda^*(c))\) denote the population KKT solution. If \(\lambda^*(c)>0\), the constraint binds and affects the solution; if \(\lambda^*(c)=0\), the constraint is slack and the estimator coincides locally with the unconstrained solution. In general, \(\theta^*(c)\) differs from the true parameter and represents a pseudo-true value.

\paragraph{Soft restrictions via constrained optimization.}
This framework departs from exact-restriction regression by treating theory as approximate rather than binding. Instead of imposing $g(\theta)=0$, we allow deviations through the constraint $g(\theta)^{\prime}\Sigma^{-1}g(\theta)\leq c$, letting the data determine the extent of enforcement. This preserves feasibility under misspecification and introduces a bias--variance tradeoff governed by $c$, where smaller values approximate exact restrictions and larger values approach the unrestricted estimator. The Lagrange multiplier provides a shadow-price interpretation, measuring the marginal gain in fit from relaxing the restriction.

\paragraph{Economic interpretation of ISP.}
The global shadow price is given by $\partial \mathcal{L}/\partial c = -\lambda$, which measures the marginal reduction in the objective from relaxing the tolerance level. The ISP decompose this quantity across restrictions via directional derivatives:
\[
\mathrm{ISP}_j(c)
=
2\lambda(c)\,\mathrm{sign}(g_j)\,[\Sigma^{-1}g]_j.
\]
Each ISP captures the marginal gain in fit from relaxing the $j$th restriction. Their magnitudes therefore provide a natural measure of empirical relevance and offer a principled basis for ranking restrictions.

\section{Asymptotic Theory}

We develop the theory in several steps. First, we show that on the active region the KKT system can be locally profiled with respect to the tolerance parameter \(c\), so that the constrained estimator and its associated shadow price vary smoothly with \(c\). Next, treating \(c\) as fixed, we derive the inner asymptotic theory for the constrained estimator, the global shadow price, and the ISP vector. We then analyze the outer problem that selects \(c\) by minimizing a Stein-type analytical risk proxy and characterize the resulting debiased optimized estimator. Finally, we introduce a plateau rule for inference on the ISP vector, which separates signal from noise using a bootstrap procedure.

\subsection{Profiling with respect to the tolerance parameter}
\label{subsec:pf_profile_c}

We first study the local profiling map with respect to the tolerance parameter \(c\). Because smooth profiling requires the equality form of the KKT system, the result is stated on the active region \(\mathcal C_{\mathrm{active}}\), where the quadratic restriction binds and the multiplier is strictly positive. This profiling result is therefore conditional on local empirical relevance of the aggregate restriction, not an assumption imposed on the whole model. For each \(c\in\mathcal C\), define the sample KKT map
\begin{equation*}
\Phi_n(\theta,\lambda;c)
:=
\begin{pmatrix}
\nabla_{\theta}\phi_n(\theta)+\lambda \nabla_{\theta}h(\theta)\\[4pt]
h(\theta)-c
\end{pmatrix},
\end{equation*}
and the population KKT map
$
\Phi(\theta,\lambda;c)
:=
\begin{pmatrix}
\nabla_{\theta}\phi(\theta)+\lambda \nabla_{\theta}h(\theta)\\[4pt]
h(\theta)-c
\end{pmatrix}.
$ Under the binding regime, the population solution \((\theta^*(c),\lambda^*(c))\) satisfies
$
\Phi(\theta^*(c),\lambda^*(c);c)=0,
\;
\lambda^*(c)>0,
$
and the sample one \((\hat\theta(c),\hat\lambda(c))\) satisfies
$
\Phi_n(\hat\theta(c),\hat\lambda(c);c)=0,
\;
\hat\lambda(c)>0.
$

\begin{assumption}[Local regularity of the profiled KKT]
\label{assump:profile_regularity}
Fix \(c^{\dagger}\in\mathcal C_{\mathrm{active}}\). Suppose
\begin{enumerate}
    \item \(\phi^0(\theta)\) is twice continuously differentiable in a neighborhood of \(\theta^*(c^{\dagger})\), and \(h(\theta)\) is twice continuously differentiable in a neighborhood of \(\theta^*(c^{\dagger})\);

    \item the active-binding regime holds at \(c^{\dagger}\):
    $
    h(\theta^*(c^{\dagger}))=c^{\dagger},
    \;
    \lambda^*(c^{\dagger})>0;
    $

    \item the Jacobian of \(\Phi(\theta,\lambda;c)\) with respect to \((\theta,\lambda)\), evaluated at \((\theta^*(c^{\dagger}),\lambda^*(c^{\dagger}),c^{\dagger})\), is nonsingular:
    $
    J_{\Phi}(c^{\dagger})
    :=
    \nabla_{(\theta,\lambda)}\Phi(\theta,\lambda;c)
    \big|_{(\theta,\lambda,c)=(\theta^*(c^{\dagger}),\lambda^*(c^{\dagger}),c^{\dagger})}
    =
    \begin{pmatrix}
    A(c^{\dagger}) & a(c^{\dagger})\\
    a(c^{\dagger})^{\prime} & 0
    \end{pmatrix},
    $
    where
    $
    a(c):=\nabla_{\theta}h(\theta^*(c)),
    \;
    A(c):=
    \nabla_{\theta\theta}^2\phi(\theta^*(c))
    +
    \lambda^*(c)\nabla_{\theta\theta}^2 h(\theta^*(c)).
    $
\end{enumerate}
\end{assumption}

\begin{lemma}[Population profiling map exists and is locally unique]
\label{lem:population_profile}
Under Assumption~\ref{assump:profile_regularity}, there exists an open neighborhood \(\mathcal N(c^{\dagger})\) of \(c^{\dagger}\) and unique continuously differentiable functions
$
c\mapsto \theta^*(c),
\;
c\mapsto \lambda^*(c),
$
defined on \(\mathcal N(c^{\dagger})\), such that
$
\Phi(\theta^*(c),\lambda^*(c);c)=0
\;\text{for all } c\in\mathcal N(c^{\dagger}),
$
with
$
h(\theta^*(c))=c,
\;
\lambda^*(c)>0.
$
\end{lemma}

\noindent \textbf{Proof.} See Appendix \ref{subsec:pf_lem-population-profile}.

\begin{assumption}[Sample regularity for profiling]
\label{assump:sample_profile}
With probability approaching 1,
\begin{enumerate}
    \item \(\phi_n(\theta)\) is twice continuously differentiable in a neighborhood of \(\theta^*(c^{\dagger})\);

    \item the sample KKT Jacobian
    $
    J_{\Phi_n}(\theta,\lambda;c)
    :=
    \nabla_{(\theta,\lambda)}\Phi_n(\theta,\lambda;c)
    =
    \begin{pmatrix}
    \nabla_{\theta\theta}^2\phi_n(\theta)+\lambda \nabla_{\theta\theta}^2 h(\theta)
    &
    \nabla_{\theta}h(\theta)\\[4pt]
    \nabla_{\theta}h(\theta)^{\prime} & 0
    \end{pmatrix}
    $
    is nonsingular uniformly over \((\theta,\lambda,c)\) in a neighborhood of \((\theta^*(c^{\dagger}),\lambda^*(c^{\dagger}),c^{\dagger})\);

    \item the sample KKT map \(\Phi_n(\theta,\lambda;c)\) converges uniformly to \(\Phi(\theta,\lambda;c)\) in that neighborhood.
\end{enumerate}
\end{assumption}

\begin{theorem}[Sample profiling map: local existence, uniqueness, and smoothness]
\label{thm:sample_profile}
Under Assumptions~\ref{assump:profile_regularity} and \ref{assump:sample_profile}, there exists a neighborhood \(\mathcal N(c^{\dagger})\) of \(c^{\dagger}\) such that, with probability approaching one, for every \(c\in\mathcal N(c^{\dagger})\), the sample KKT system
$
\Phi_n(\theta,\lambda;c)=0
$
admits a unique local solution
$
(\hat\theta(c),\hat\lambda(c))
$
satisfying
$
h(\hat\theta(c))=c,
\;
\hat\lambda(c)>0.
$
Moreover, the map
$
c\mapsto (\hat\theta(c),\hat\lambda(c))
$ 
is continuously differentiable on \(\mathcal N(c^{\dagger})\) with probability approaching one.
\end{theorem}

\noindent \textbf{Proof.} See Appendix \ref{subsec:pf_thm-sample-profile}.

\begin{corollary}[Local validity of profiling with respect to \(c\)]
\label{cor:profiling_valid}
Under Assumptions~\ref{assump:profile_regularity} and \ref{assump:sample_profile}, profiling with respect to the tolerance parameter is locally well defined on the active region. In particular, for all \(c\) in a neighborhood of \(c^{\dagger}\), the constrained estimator can be written as
$
(\hat\theta(c),\hat\lambda(c)),
$
where both coordinates are uniquely determined and vary smoothly with \(c\). Consequently, optimization of any smooth criterion over \(c\), including the analytical Stein-type proxy \(\widehat R_{\mathrm{lin}}(c)\), is locally well posed in that region.
\end{corollary}

\noindent \textbf{Proof.} See Appendix \ref{subsec:pf_cor-profiling-valid}.

\subsection{Inner asymptotic theory for fixed \(c\)}

We now treat \(c\) as fixed and derive the inner asymptotic theory for the constrained estimator, the global shadow price, and the ISP vector. The theory is stated without imposing a global positive-multiplier assumption. Instead, whenever a result requires a binding quadratic restriction, this is stated explicitly by restricting attention to \(c\in\mathcal C_{\mathrm{active}}\).

\begin{assumption}[Slater's condition / constraint qualification]
\label{assump:slater}
For each $c \in \mathcal C$, the feasible set
$
\Theta(c) := \{\theta \in \Theta : h(\theta) \le c\}
$
is nonempty, and there exists $\bar{\theta}(c) \in \Theta$ such that
$
h(\bar{\theta}(c)) < c.
$
\end{assumption}

\begin{assumption}\label{assump:smooth_cvx}
\textbf{(Smoothness and convexity)}
For each \(c\in\mathcal C\), \(\phi(z;\theta)\) is twice continuously differentiable in \(\theta\), and the population objective
$
\phi^0\equiv \mathbb E[\phi(Z_i;\theta^{0})]
$
is strictly convex in \(\theta\) on a neighborhood of \(\theta^*(c)\).
\end{assumption}

\begin{assumption}\label{assump:regular}
\textbf{(Restriction regularity)}
For each \(c\in\mathcal C\), \(g(\theta)\) is continuously differentiable on a neighborhood of \(\theta^*(c)\), and the Jacobian
$
G(\theta)\equiv \nabla_\theta g(\theta)
$
has full row rank \(q\) at \(\theta^*(c)\).
\end{assumption}

\begin{definition}[Active and inactive regions]
\label{def:active_region}
Define
$
\mathcal C_{\mathrm{active}}
:=
\{c\in\mathcal C:\lambda^*(c)>0\},
\;
\mathcal C_{\mathrm{inactive}}
:=
\{c\in\mathcal C:\lambda^*(c)=0\}.
$
\end{definition}

\begin{assumption}\label{assump:S_CLT}
\textbf{(Uniform laws of large numbers and central limit theorem)}
Let \(\psi_i(\theta)\equiv \nabla_\theta \phi(Z_i;\theta)\) denote the score contribution and
$
\psi_n(\theta)\equiv n^{-1}\sum_{i=1}^n \psi_i(\theta)
$
the sample average score with
$
\psi_0\equiv \mathbb E[\psi_i(\theta^{0})].
$ 

For each \(c\in\mathcal C_{\mathrm{active}}\), assume that
$
\nabla_\theta \phi(\theta^*(c))
+
\lambda^*(c)\nabla_\theta h(\theta^*(c))=0,
$
and \\
$
\sqrt{n}\,\psi_n(\theta^*(c)) \xrightarrow{d} N(0,\Sigma_{\psi}(c)),
\;
\Sigma_{\psi}(c)\equiv \mathrm{Var}\!\big(\psi_i(\theta^*(c))\big)
$, and there exists a neighborhood \(\mathcal N(c)\) of \(\theta^*(c)\) where uniform laws of large numbers hold for \(\phi_n(\theta)\) and its first two derivatives:
$
\sup_{\theta\in\mathcal N(c)}
\big|\phi_n(\theta)-\phi(\theta)\big| \xrightarrow{p} 0, \;
\sup_{\theta\in\mathcal N(c)}
\big\|\nabla_\theta \phi_n(\theta)-\nabla_\theta \phi(\theta)\big\| \xrightarrow{p} 0, \;
\sup_{\theta\in\mathcal N(c)}
\big\|\nabla_{\theta\theta}^2 \phi_n(\theta)-\nabla_{\theta\theta}^2 \phi(\theta)\big\| \xrightarrow{p}0.$
\end{assumption}

Assumption~\ref{assump:slater} ensures that the feasible set has a nonempty interior and that the KKT conditions characterize the constrained optimum. Assumption~\ref{assump:smooth_cvx} ensures existence and uniqueness of the population constrained minimizer and justifies a second-order expansion of the Lagrangian. Assumption~\ref{assump:regular} guarantees local regularity of the constraint surface and identification of the KKT system. Assumption~\ref{assump:slater} ensures strong duality and validity of the KKT characterization. Definition~\ref{def:active_region} distinguishes the active and inactive regimes. On the active region, the shadow-price system is nondegenerate and admits a smooth first-order expansion. On the inactive region, the quadratic restriction is slack and the constrained estimator reduces locally to the unconstrained regime. Accordingly, the asymptotic results below for \((\hat\theta(c),\hat\lambda(c))\) and the ISP vector are stated for fixed \(c\in\mathcal C_{\mathrm{active}}\). Assumption~\ref{assump:S_CLT} provides the stochastic equicontinuity and central limit theorem required for a first-order expansion of the score around \(\theta^*(c)\). The uniform laws of large numbers for \(\phi_n(\theta)\) and its derivatives guarantee local stability of the objective and its curvature, allowing a valid Taylor expansion in a neighborhood of \(\theta^*(c)\).

For each fixed \(c\in\mathcal C_{\mathrm{active}}\), define
$
a(c)\equiv \nabla_\theta h(\theta^*(c))
=
2\,G(\theta^*(c))^{\prime}\Sigma^{-1}g(\theta^*(c)),
$
and the population Hessian of the Lagrangian,
$
A(c)\equiv \nabla^2_{\theta\theta}\mathcal L(\theta^*(c),\lambda^*(c);c)
=
H(\theta^*(c))+\lambda^*(c)\nabla^2_{\theta\theta}h(\theta^*(c)),
$ with
$H(\theta^*(c))\equiv \nabla^2_{\theta\theta}\phi(\theta^*(c))$.
Here \(a(c)\) is the gradient of the binding restriction and therefore the normal vector to the constraint surface \(\{\theta:h(\theta)=c\}\) at \(\theta^*(c)\). As a normal vector, \(a(c)\) identifies the direction in parameter space along which movement is locally forbidden when the constraint binds: moving in the direction of \(a(c)\) or \(-a(c)\) changes the value of \(h(\theta)\) to first order and hence relaxes or tightens the restriction. In contrast, perturbations orthogonal to \(a(c)\) lie in the tangent space of the constraint surface and preserve feasibility to first order. 

The matrix \(A(c)\) is the curvature matrix of the population Lagrangian with respect to \(\theta\), evaluated at the KKT solution. Since \(\theta^*(c)\) is assumed to be a local minimizer of the constrained loss, standard second-order sufficient conditions imply that \(A(c)\) is positive definite, ensuring local uniqueness and identification of the solution. Under these conditions, \(A(c)^{-1}\) exists and is positive definite, and therefore \(a(c)^{\prime}A(c)^{-1}a(c)>0\) for any nonzero \(a(c)\).

\begin{theorem}[Asymptotic normality under a binding constraint with optimized tolerance]
\label{thm:asymp_theta_active}
Under Assumptions \ref{assump:slater}, \ref{assump:smooth_cvx}, \ref{assump:regular},  and \ref{assump:S_CLT}, for each \(c\in\mathcal C_{\mathrm{active}}\),
\begin{equation}\label{eq:joint_CLT}
\sqrt{n}
\begin{pmatrix}
\hat{\theta}(c)-\theta^*(c)\\
\hat{\lambda}(c)-\lambda^*(c)
\end{pmatrix}
\xrightarrow{d}
\mathcal{N}\!\left(
0,\begin{pmatrix}
V_1(c) & V_2(c)\\
V_2(c)^{\prime} & V_3(c)
\end{pmatrix}
\right), \quad \text{where}
\end{equation} \vspace{-10mm}
\begin{align*}
V_1(c)
&\equiv \mathrm{Var}\!\big(\sqrt{n}(\hat\theta(c)-\theta^*(c))\big)
      = B_\theta(c)\,\Sigma_{\psi}(c)\,B_\theta(c)^{\prime}
      = M(c)\,\Sigma_{\psi}(c)\,M(c)^{\prime},\\
V_2(c)
&\equiv \mathrm{Cov}\!\big(\sqrt{n}(\hat\theta(c)-\theta^*(c)),\,\sqrt{n}(\hat\lambda(c)-\lambda^*(c))\big)
      = B_\theta(c)\,\Sigma_{\psi}(c)\,B_\lambda(c)^{\prime} \\
&= -\,M(c)\,\Sigma_{\psi}(c)\,A(c)^{-1}a(c)\,(a(c)^{\prime}A(c)^{-1}a(c))^{-1},\\
V_3(c)
&\equiv \mathrm{Var}\!\big(\sqrt{n}(\hat\lambda(c)-\lambda^*(c))\big)
      = B_\lambda(c)\,\Sigma_{\psi}(c)\,B_\lambda(c)^{\prime} \\
&= (a(c)^{\prime}A(c)^{-1}a(c))^{-2}\,a(c)^{\prime}A(c)^{-1}\Sigma_{\psi}(c) A(c)^{-1}a(c),
\end{align*}
with
$
B_\theta(c)\equiv -M(c),
\;
M(c)\equiv A(c)^{-1}-A(c)^{-1}a(c)(a(c)^{\prime}A(c)^{-1}a(c))^{-1}a(c)^{\prime}A(c)^{-1},
$
and
$
B_\lambda(c)\equiv -\,(a(c)^{\prime}A(c)^{-1}a(c))^{-1}a(c)^{\prime}A(c)^{-1}.
$
\end{theorem}

\noindent \textbf{Proof.} See Appendix \ref{subsec:derivation_asymp}.

\begin{theorem}[Asymptotic normality under a non-binding constraint with optimized tolerance]
\label{thm:asymp_theta_inactive}
Suppose Assumptions \ref{assump:slater}, \ref{assump:smooth_cvx}, \ref{assump:regular}, and \ref{assump:S_CLT} hold. Let
$
\hat{\theta} := \hat{\theta}(\hat c),
\;
\hat{\lambda} := \hat{\lambda}(\hat c),
$
where $\hat c$ is the data-driven tolerance parameter and $c^{\dagger}$ denotes its probability limit, the pseudo-true optimal tolerance. Let $(\theta^*(c),\lambda^*(c))$ denote the population KKT solution for a given $c$.

If the population constraint is strictly slack at $c^{\dagger}$: $
h(\theta^*(c^{\dagger})) < c^{\dagger}
\;\text{and}\;
\lambda^*(c^{\dagger}) = 0,
$
then the constraint is asymptotically non-binding at the optimal tolerance. In this case,
$
\sqrt{n}\big(\hat{\theta}-\theta^*(c^{\dagger})\big)
\xrightarrow{d}
N\!\Bigl(0,\ H(\theta^*(c^{\dagger}))^{-1}\Sigma_{\psi}H(\theta^*(c^{\dagger}))^{-1}\Bigr),
$
where $H(\theta)=\nabla_{\theta\theta}^2\phi(\theta)$ and $\Sigma_{\psi}=\mathrm{Var}(\psi(Z_i;\theta^*(c^{\dagger})))$. Moreover,
$
\hat{\lambda}\xrightarrow{p}0.
$
\end{theorem}

\noindent \textbf{Proof.} See Appendix \ref{subsec:pf_asymp_inactive}.

The matrix \(M(c)\) can be interpreted as a projected inverse Hessian. It modifies the unconstrained influence matrix \(A(c)^{-1}\) by removing variation in the direction that would locally violate the binding constraint. This direction is identified by the constraint gradient \(a(c)=\nabla_\theta h(\theta^*(c))\), which is normal to the constraint surface. Variation along \(a(c)\) would move the estimator off the feasible set and is therefore ruled out by the binding restriction. The projection implemented by \(M(c)\) eliminates this forbidden component, ensuring that the asymptotic variation of \(\hat{\theta}(c)\) lies entirely in directions that preserve feasibility to first order.

\subsection{Asymptotic distribution of ISPs}

For each \(c\in\mathcal C_{\mathrm{active}}\), define the population ISP:
$ISP_j^*(c)
:=
2\lambda^*(c)\,\mathrm{sign}\!\big(g_j(\theta^*(c))\big)\,
\big[\Sigma^{-1}g(\theta^*(c))\big]_j,
$
and the sample ISP estimator
$
\widehat{ISP}_j(c)
:=
2\hat\lambda(c)\,\mathrm{sign}\!\big(g_j(\hat\theta(c))\big)\,
\big[\Sigma^{-1}g(\hat\theta(c))\big]_j.
$
Stack
$
\widehat{ISP}(c)\equiv(\widehat{ISP}_1(c),\dots,\widehat{ISP}_q(c))^{\prime},
\;
ISP^*(c)\equiv(ISP_1^*(c),\dots,ISP_q^*(c))^{\prime}.
$

\begin{assumption}[Regularity for the sign map]
\label{assump:sign_map}
For the set of restriction indices \(\mathcal J \subset \{1,\dots,q\}\) used to construct and analyze the ISP, and for each fixed \(c\in\mathcal C_{\mathrm{active}}\),
\begin{equation*}
\min_{j \in \mathcal J} |g_j(\theta^*(c))| \ge \kappa(c), \quad \text{for some constant \(\kappa(c)>0\) independent of \(n\).}
\end{equation*}
\end{assumption}

Assumption~\ref{assump:sign_map} guarantees that the sign of each relevant restriction is locally stable in a neighborhood of the population target \(\theta^*(c)\). Hence, with probability approaching one, the sign operator is locally constant and does not affect first-order delta-method expansions of the ISP. If \(g_j(\theta^*(c))=0\) for some \(j\), the sign map is non-differentiable at \(\theta^*(c)\), and stochastic sign switching may occur in \(n^{-1/2}\) neighborhoods, leading to nonregular asymptotics that require separate analysis. Note that Assumption~\ref{assump:sign_map} is imposed at $\theta^*(c)$, rather than at $\theta^0$. Since ISP is evaluated at a tolerance level $c$ that optimizes the bias--variance tradeoff and is therefore not chosen to recover $\theta^0$, a restriction need not hold exactly at the pseudo-true constrained target $\theta^*(c)$, even if it holds exactly at the true data-generating parameter $\theta^0$.

\begin{theorem}[Joint asymptotic distribution of the ISPs]
\label{thm:joint_ISP}
Under Assumptions \ref{assump:slater}, \ref{assump:smooth_cvx}, \ref{assump:regular},  \ref{assump:S_CLT}, and \ref{assump:sign_map}, for each fixed \(c\in\mathcal C_{\mathrm{active}}\),
\begin{equation}\label{eq:ISP_joint_CLT}
\sqrt{n}\big(\widehat{ISP}(c)-ISP^*(c)\big)
\xrightarrow{d}
\mathcal{N}\!\big(0,\Sigma_{ISP}(c)\big),
\end{equation}
where
$
\Sigma_{ISP}(c)
=
J_\theta(c) V_1(c) J_\theta(c)^{\prime}
+
J_\theta(c) V_2(c) J_\lambda(c)^{\prime}
+
J_\lambda(c) V_2(c)^{\prime} J_\theta(c)^{\prime}
+
J_\lambda(c) V_3(c) J_\lambda(c)^{\prime},
$
with
$
J_\theta(c)
\equiv
\left.\nabla_\theta ISP(\theta,\lambda;c)\right|_{(\theta^*(c),\lambda^*(c))}
=
2\lambda^*(c)\,S(c)\,\Sigma^{-1}G(\theta^*(c))
\in\mathbb{R}^{q\times p},\;\\
J_\lambda(c)
\equiv
\left.\nabla_\lambda ISP(\theta,\lambda;c)\right|_{(\theta^*(c),\lambda^*(c))}
=
2\,S(c)\,\Sigma^{-1}g(\theta^*(c))
\in\mathbb{R}^{q\times 1},
$
and\\
$
S(c)\equiv \mathrm{diag}\!\big(\mathrm{sign}(g_1(\theta^*(c))),\dots,\mathrm{sign}(g_q(\theta^*(c)))\big).
$
\end{theorem}

\noindent \textbf{Proof.} See Appendix \ref{subsec:pf_ISP-joint}.

\begin{remark}[Binding-regime interpretation and implications for ranking]
The joint asymptotic distribution of the ISP vector reveals two distinct channels through which sampling uncertainty enters the ranking system under an active constraint. First, the presence of the common factor \(\lambda^*(c)>0\) in \(J_\theta(c)\) implies that variation in \(\hat\theta(c)\) affects all ISPs through the credibility-weighted Jacobian \(\Sigma^{-1}G(\theta^*(c))\) on a common global scale. Second, \(J_\lambda(c)\) propagates uncertainty in the global shadow price \(\hat\lambda(c)\) uniformly across restrictions in proportion to \(\Sigma^{-1}g(\theta^*(c))\). Because both channels operate within a single shadow-price system, relative differences across ISPs arise from restriction-specific features rather than heterogeneous scaling.
\end{remark}

\subsection{ISP ranking and cutoff}

\begin{lemma}[Uniform convergence of ISP]
\label{lem:ISP_uniform}
Suppose the number of restrictions \(q\) is fixed and Theorem~\ref{thm:joint_ISP} holds for a given \(c\in\mathcal C_{\mathrm{active}}\). Then
\begin{equation}
\max_{1 \le j \le q}
\left|
\widehat{ISP}_j(c)-ISP_j^*(c)
\right|
\;\xrightarrow{p}\; 0.
\end{equation}
\end{lemma}

\noindent\textbf{Proof.} See Appendix~\ref{subsec:pf_ISP-uniform}.

To formalize ranking in terms of magnitudes, let \(r^*(c)\) denote the permutation that sorts population ISP magnitudes in weakly decreasing order:
\begin{equation*}
|ISP_{r^*(c;1)}^*(c)|
\ge
|ISP_{r^*(c;2)}^*(c)|
\ge
\cdots
\ge
|ISP_{r^*(c;q)}^*(c)|.
\end{equation*}

\begin{assumption}[Separation of population ISP magnitudes]
\label{assump:ISP_separation}
For a given \(c\in\mathcal C_{\mathrm{active}}\), there exists \(\Delta(c)>0\) such that for all \(k\in\{1,\dots,q-1\}\),
\begin{equation}
|ISP_{r^*(c;k)}^*(c)|
-
|ISP_{r^*(c;k+1)}^*(c)|
\ge
\Delta(c).
\end{equation}
\end{assumption}

\begin{theorem}[Uniform validity of ISP ranking]
\label{thm:ISP_ranking_consistency}
Under Assumptions \ref{assump:slater},
\ref{assump:smooth_cvx},
\ref{assump:regular},
\ref{assump:S_CLT},
\ref{assump:sign_map},
and \ref{assump:ISP_separation},
let \(\hat r(c)\) denote the permutation that sorts the estimated ISP magnitudes in weakly decreasing order:
$
|\widehat{ISP}_{\hat r(c;1)}(c)|
\ge
|\widehat{ISP}_{\hat r(c;2)}(c)|
\ge
\cdots
\ge
|\widehat{ISP}_{\hat r(c;q)}(c)|.
$
Then, for each fixed \(c\in\mathcal C_{\mathrm{active}}\),
$
\Pr\bigl(\hat r(c)=r^*(c)\bigr)\to 1.
$
\end{theorem}

\noindent\textbf{Proof.} See Appendix~\ref{subsec:pf_isp-rank}.

\paragraph{Plateau cutoff \(\hat m(c)\).}
Building on the ordered ISP magnitudes, we define a point estimator of the boundary between noise and signal restrictions. Let
$
|\widehat{ISP}_{(1)}(c)| \le \cdots \le |\widehat{ISP}_{(q)}(c)|
$
denote the ordered absolute ISPs, where the ordering is now from smallest to largest so that the lower tail corresponds to candidate noise restrictions. For each candidate cutoff \(m\in\{m_0,\dots,q-1\}\), define
$
\mathcal G_1(m)=\{1,\dots,m\},
\;
\mathcal G_2(m)=\{m+1,\dots,q\},
$
where \(\mathcal G_1(m)\) is interpreted as a candidate plateau and \(\mathcal G_2(m)\) as the remaining signal region. We consider the contrast
\[
\Delta_m(c)
=
\frac{1}{m}\sum_{\ell=1}^m |\widehat{ISP}_{(\ell)}(c)|
-
\frac{1}{q-m}\sum_{\ell=m+1}^q |\widehat{ISP}_{(\ell)}(c)|.
\]
To assess whether a structural break occurs at \(m\), we test
\[
H_0:\ \Delta_m(c)=0,
\]
using the Wald-type statistic
$
B_m(c)
=
n\,\Delta_m(c)^2 \big/ \widehat{\mathrm{Var}}(\Delta_m(c)),
$
where \(\widehat{\mathrm{Var}}(\Delta_m(c))\) is obtained from \(\widehat{\Sigma}_{ISP}(c)\) via a linear contrast.

To ensure that the lower block behaves like a plateau, we also impose a within-block homogeneity screen. For each candidate \(m\), we test
\[
H_0:\ 
|\widehat{ISP}_{(1)}(c)|
=
\cdots
=
|\widehat{ISP}_{(m)}(c)|,
\]
using a joint Wald test based on pairwise contrasts within \(\mathcal G_1(m)\). Let \(\mathcal M_{\mathrm{adm}}\) denote the set of candidate cutoffs that pass this homogeneity screen. We then define the point estimator
$
\hat m(c)
=
\arg\max_{m\in\mathcal M_{\mathrm{adm}}}
B_m(c),
$
with the maximization carried out over all feasible candidates if no \(m\) passes the screen. The estimated plateau is given by \(\{1,\dots,\hat m(c)\}\), and restrictions beyond \(\hat m(c)\) are classified as signal.

\paragraph{Population counterpart.}
Let
$
|ISP_{(1)}^*(c)| \le \cdots \le |ISP_{(q)}^*(c)|
$
denote the ordered population ISP magnitudes, and define
$
\Delta_m^*(c)
=
\frac{1}{m}\sum_{\ell=1}^m |ISP_{(\ell)}^*(c)|
-
\frac{1}{q-m}\sum_{\ell=m+1}^q |ISP_{(\ell)}^*(c)|.
$

\begin{assumption}[Population structural break]
\label{ass:plateau_sep}
For a fixed \(c\in\mathcal C_{\mathrm{active}}\), there exists a unique index \(m^*(c)\in\{m_0,\dots,q-1\}\) and a constant \(\delta(c)>0\) such that
\[
\Delta_m^*(c)\le \Delta_{m^*(c)}^*(c)-\delta(c)
\qquad
\text{for all } m\neq m^*(c).
\]
\end{assumption}

Assumption~\ref{ass:plateau_sep} states that the ordered population ISP sequence admits a well-defined structural break separating a lower-magnitude noise region from a higher-magnitude signal region, in the sense that the population break criterion has a unique maximizer.

\begin{proposition}[Consistency of plateau recovery]
\label{prop:plateau_consistency}
Fix \(c\in\mathcal C_{\mathrm{active}}\). Suppose Theorem~\ref{thm:joint_ISP} holds, \(q\) is fixed, and Assumption~\ref{ass:plateau_sep} is satisfied. Let \(\hat m(c)\) denote the estimated cutoff obtained from the plateau rule above. Then
$
\hat m(c)\xrightarrow{p} m^*(c).
$
\end{proposition}

\noindent\textbf{Proof.} See Appendix \ref{subsec:pf_prop-plateau-consistency}.

\begin{remark}
The ISP ranking theory remains an inner-level result: it applies conditionally on a fixed tolerance value \(c\) in the active set. In the optimized procedure, it therefore applies in particular at the pseudo-true optimal tolerance \(c^\dagger\) and, under consistency of \(\hat c\), asymptotically at the selected tolerance \(\hat c\), provided \(c^\dagger\in\mathcal C_{\mathrm{active}}\). The same applies to the point estimator \(\hat m(c)\), whose consistency is established under the structural-break condition in Assumption~\ref{ass:plateau_sep}.
\end{remark}

\subsection{Analytical Stein-type risk optimization over \(c\)}

We now turn to the outer problem and endogenize the tolerance level by selecting \(c\) through an analytical Stein-type risk criterion. Let \(\theta^{0}\) denote the unconstrained target parameter defined by
$
\nabla_\theta \phi^0(\theta^{0})=0.
$
For each \(c\in\mathcal C_{\mathrm{active}}\), define the population risk
$
R(c):=\mathbb E\big[\|\hat\theta(c)-\theta^{0}\|^2\big].
$
Using the decomposition
\[
\hat\theta(c)-\theta^{0}
=
\big(\hat\theta(c)-\theta^*(c)\big)
+
\big(\theta^*(c)-\theta^{0}\big),
\]
the leading terms of the risk are a variance component and a regularization-bias component. 

\begin{proposition}[Local expansion of the pseudo-true parameter]
\label{prop:theta_star_expansion}
Suppose Assumptions~\ref{assump:slater} , \ref{assump:smooth_cvx}, and \ref{assump:regular} hold, and let \(\theta^{0}\) denote the unconstrained benchmark satisfying
$
\nabla_\theta \phi(\theta^{0})=0.
$
For each \(c\in\mathcal C_{\mathrm{active}}\), let \((\theta^*(c),\lambda^*(c))\) denote the population KKT solution. Then, in a neighborhood of \(\theta^{0}\),
\begin{equation}
\theta^*(c)-\theta^{0}
=
-\lambda^*(c)H_0^{-1}a_0
+
r(c),
\label{eq:theta_star_expansion}
\end{equation}
where
$
H_0=\nabla_{\theta\theta}^2\phi(\theta^{0}),
\;
a_0=\nabla_\theta h(\theta^{0}),
$
and the remainder satisfies
$
\|r(c)\|=o\!\big(\lambda^*(c)\big).
$
\end{proposition}

\noindent \textbf{Proof.} See Appendix \ref{subsec:pf_prop-theta-star-expansion}.

Proposition \ref{prop:theta_star_expansion} implies that the leading population risk approximation is
\begin{equation}
R_{\mathrm{lin}}(c)
=
\lambda^*(c)^2 a_0^{\prime}H_0^{-1\prime}H_0^{-1}a_0
+
\frac{1}{n}\operatorname{tr}\big(V_1(c)\big).
\label{eq:R_lin_pop}
\end{equation}
Motivated by \eqref{eq:R_lin_pop}, define the feasible analytical proxy
\begin{equation}
\widehat R_{\mathrm{lin}}(c)
=
\widehat B_{\mathrm{lin}}(c)+\widehat W(c),
\label{eq:R_lin_hat}
\end{equation}
where
$
\widehat W(c)
=
\frac{1}{n}\operatorname{tr}\big(\widehat V_1(c)\big),
$
and
$
\widehat B_{\mathrm{lin}}(c)
=
\hat{\lambda}(c)^2\,
\hat a^{\prime}\hat H^{-1\prime}\hat H^{-1}\hat a.
$
Here \(\tilde\theta\) is the unconstrained target estimator
$
\tilde\theta
=
\arg\min_{\theta}\phi_n(\theta),
$
and
$
\hat a=\nabla_\theta h(\tilde\theta),
\;
\hat H=\nabla_{\theta\theta}^2\phi_n(\tilde\theta).
$ 

The tolerance level is selected as
$
\hat c \in \arg\min_{c \in \mathcal C_{\mathrm{active}}} \widehat R_{\mathrm{lin}}(c).
$
Let the pseudo-true optimal tolerance be
$
c^{\dagger}
=
\arg\min_{c \in \mathcal C_{\mathrm{active}}} R_{\mathrm{lin}}(c).
\label{eq:c_dagger_def}
$

\begin{assumption}[Identification and uniform consistency of the risk criterion]
\label{assump:c_selection}
The set \(\mathcal C_{\mathrm{active}}\) is compact, the population criterion \(R_{\mathrm{lin}}(c)\) is continuous on \(\mathcal C_{\mathrm{active}}\), and it has a unique minimizer \(c^{\dagger}\). Moreover,
$
\sup_{c\in\mathcal C_{\mathrm{active}}}
\big|
\widehat R_{\mathrm{lin}}(c)-R_{\mathrm{lin}}(c)
\big|
\;\xrightarrow{p}\;0.
$
\end{assumption}

\begin{theorem}[Consistency of the optimized tolerance]
\label{thm:c_hat_consistency}
Under Assumptions \ref{assump:slater}, \ref{assump:smooth_cvx}, \ref{assump:regular},  \ref{assump:S_CLT}, and \ref{assump:c_selection},
$
\hat c \xrightarrow{p} c^{\dagger}.
$
\end{theorem}

\noindent \textbf{Proof.} See Appendix \ref{subsec:pf_thm-c-hat-consistency}.

\begin{corollary}[Consistency of the optimized estimator]
\label{cor:theta_hat_opt_consistency}
Under the assumptions of Theorem~\ref{thm:c_hat_consistency}, if the maps
$
c\mapsto \theta^*(c),
\;
c\mapsto \lambda^*(c)
$
are continuous at \(c^{\dagger}\), then
\[
\hat{\theta}=\hat{\theta}(\hat c)\xrightarrow{p}\theta^*(c^{\dagger}),
\qquad
\hat{\lambda}=\hat{\lambda}(\hat c)\xrightarrow{p}\lambda^*(c^{\dagger}).
\]
\end{corollary}

\noindent \textbf{Proof.} See Appendix \ref{subsec:pf_cor-theta-hat-opt-consistency}.

\begin{remark}[Interpretation of the optimized tolerance]
The optimized tolerance scheme balances two forces. Smaller values of \(c\) impose the restrictions more tightly and may reduce variance, but they also increase regularization bias when the restrictions are misspecified. Larger values of \(c\) relax the restriction set, thereby reducing regularization bias, but may increase variance. The criterion \(\widehat R_{\mathrm{lin}}(c)\) formalizes this tradeoff and selects the tolerance level that is optimal, within the admissible active set, according to the analytical Stein-type approximation.
\end{remark}

\subsection{Debiasing after optimized tolerance selection}
\label{subsec:debias_after_c_opt}
The optimization over \(c\) via the Stein-type risk criterion selects a tolerance level that optimally balances regularization bias and variance, thereby determining the degree of bias that is acceptable in minimizing overall estimation risk. Consequently, the estimator \(\hat\theta(\hat c)\) is tuned to achieve an optimal bias--variance tradeoff.

Nevertheless, this optimization operates at a global level and does not, in general, eliminate the bias in the centering of the estimator. While the optimized tolerance parameter controls the bias--variance tradeoff, it does not remove the regularization bias induced by the constraint. A subsequent debiasing step is therefore required to eliminate the remaining bias and restore correct centering for inference. In particular, even at the optimal tolerance \(c^\dagger\), we typically have \(\theta^*(c^\dagger)\neq \theta^{0}\) whenever the quadratic restriction remains locally binding.

\paragraph{Local characterization of the regularization bias.} Proposition~\ref{prop:theta_star_expansion} shows that the leading regularization bias is of order \(O(\lambda^*(c))\), with direction determined by the curvature of the objective and the gradient of the constraint. Thus, the Lagrange multiplier governs both the strength and direction of the distortion induced by the constraint. 
Since
$
\theta^*(c)-\theta^{0} = O\big(\lambda^*(c)\big),
$
the leading bias is proportional to \(\lambda^*(c)\), which captures the local tightness of the constraint. This characterization operates at the level of the parameter \(\theta\), providing a tractable approximation to the residual bias after optimization over \(c\), and thereby enables a simple debiasing correction that recenters the estimator toward \(\theta^{0}\) without affecting its first-order variance.

Motivated by \eqref{eq:theta_star_expansion}, define the population bias proxy
$
b^*(c):=\lambda^*(c)H_0^{-1}a_0.
$
Its feasible analogue is
$
\hat b(c):=\hat\lambda(c)\,\hat H^{-1}\hat a,
\;
\hat a=\nabla_\theta h(\tilde\theta),
\;
\hat H=\nabla_{\theta\theta}^2\phi_n(\tilde\theta),
$
where \(\tilde\theta\) is the unconstrained estimator. The debiased optimized estimator is then
\begin{equation}
\hat\theta_{\mathrm{db}}
:=
\hat\theta(\hat c)+\hat b(\hat c).
\label{eq:theta_db_def}
\end{equation}
The sign convention in \eqref{eq:theta_db_def} follows from
$
\theta^*(c)-\theta^{0}\approx -\,b^*(c),
$
so that adding \(\hat b(\hat c)\) removes the leading \(O(\lambda)\) regularization bias.

\begin{assumption}[Bias approximation and smoothness of the debiasing map]
\label{assump:debias}\hspace{5mm}
    1. The linearized bias approximation holds uniformly on \(\mathcal C_{\mathrm{active}}\):
    
    $
    \sup_{c\in\mathcal C_{\mathrm{adm}}}
    \left\|
    \theta^*(c)-\theta^{0}+\lambda^*(c)H_0^{-1}a_0
    \right\|
    =
    o(n^{-1/2});
    $
    
\noindent 2. the maps
    $
    c\mapsto \theta^*(c),\;
    c\mapsto \lambda^*(c),\;
    c\mapsto A(c),\;
    c\mapsto a(c)
    $
    are continuously differentiable in a neighborhood of \(c^\dagger\);

\noindent 3. \(\hat H\xrightarrow{p} H_0\), \(\hat a\xrightarrow{p} a_0\), and
    $
    \sup_{c\in\mathcal C_{\mathrm{active}}}
    \|\hat b(c)-b^*(c)\|
    =
    o_p(1).
    $

\end{assumption}

\begin{theorem}[Asymptotic expansion of the debiased optimized estimator]
\label{thm:theta_db_asymp}
Under Assumptions \ref{assump:slater},
\ref{assump:smooth_cvx},
\ref{assump:regular},
\ref{assump:S_CLT},
\ref{assump:c_selection},
and \ref{assump:debias},
\begin{equation}
\sqrt n\bigl(\hat\theta_{\mathrm{db}}-\theta^{0}\bigr)
=
\sqrt n\bigl(\hat\theta(\hat c)-\theta^*(c^\dagger)\bigr)
+
\sqrt n\bigl(\hat b(\hat c)-b^*(c^\dagger)\bigr)
+
o_p(1),
\label{eq:theta_db_expansion}
\end{equation}
where
$
b^*(c):=\lambda^*(c)H_0^{-1}a_0,
\;
\hat b(c):=\hat\lambda(c)\hat H^{-1}\hat a.
$
If, in addition, \(\hat c-c^\dagger=o_p(n^{-1/2})\), \begin{equation}
\sqrt n\bigl(\hat\theta_{\mathrm{db}}-\theta^{0}\bigr)
\xrightarrow{d}
N\!\bigl(0,V_{\mathrm{db}}(c^\dagger)\bigr),
\label{eq:theta_db_clt}
\end{equation}
where \(V_{\mathrm{db}}(c^\dagger)\) is the asymptotic variance induced by the joint
first-order influence of
$
\hat\theta(c^\dagger),\;
\hat\lambda(c^\dagger),\;
\hat H,\;
\hat a.
$
\end{theorem}
\noindent \textbf{Proof.} See Appendix \ref{subsec:pf_theta-db-asymp}.

\begin{remark}[Interpretation]
Debiasing is needed only in the active regime. If the global shadow price is zero and
the aggregate quadratic restriction is slack, then \(\hat\theta(\hat c)\) is asymptotically
equivalent to the unconstrained estimator, and the leading regularization bias
vanishes. In that case,
$
\hat b(\hat c)=o_p(n^{-1/2}).
$
\end{remark}

\section{Extension: bootstrap-based inference}
\label{sec:wild_bootstrap}

As an extension, we develop a bootstrap-based implementation of the inference procedure. The bootstrap is formally justified for the smooth objects in the system, namely \((\hat\theta(c),\hat\lambda(c))\) and the ISP vector. It is then used to propagate sampling variability through the plateau-selection rule and to construct an empirical distribution of the cutoff. This is especially important because the separation condition in Assumption~\ref{assump:ISP_separation} is not directly verifiable in practice and may be weak when several restrictions have similar empirical effects.

\subsection{Wild bootstrap inference for the optimized estimator and ISP}

Let \(\hat c\) denote the optimized tolerance selected by the outer optimization step, and let \((\hat\theta(\hat c),\hat\lambda(\hat c))\) be the corresponding constrained estimator and shadow price from the inner problem. Define the residual
$
\hat u_i(\hat c)=y_i-x_i^{\prime}\hat\theta(\hat c),
$
evaluated at the constrained estimator associated with the selected tolerance. For each bootstrap replication \(b=1,\dots,B\), draw i.i.d.\ multipliers \(\{w_i^{(b)}\}_{i=1}^n\) such that
$
\mathbb E[w_i]=0,\;
\mathbb E[w_i^2]=1,\;
\mathbb E\!\left[|w_i|^{2+\nu}\right]<\infty
\ \text{for some }\nu>0,
$
independently of the data. Construct pseudo-outcomes by perturbing the residual component while holding regressors fixed:
\[
y_i^{(b)}=x_i^{\prime}\hat\theta(\hat c)+\hat u_i(\hat c)\,w_i^{(b)}.
\]

For each bootstrap sample, we then re-implement the full estimation procedure. First, using \(\{(y_i^{(b)},x_i)\}_{i=1}^n\), we solve the outer optimization problem to obtain a bootstrap-selected tolerance \(\hat c^{(b)}\). Second, conditional on \(\hat c^{(b)}\), we solve the corresponding inner constrained problem to obtain \((\hat\theta^{(b)}(\hat c^{(b)}),\hat\lambda^{(b)}(\hat c^{(b)}))\):
\[
(\hat\theta^{(b)}(\hat c^{(b)}),\hat\lambda^{(b)}(\hat c^{(b)}))
\in
\arg\min_{\theta\in\mathbb R^p}\ \phi_n^{(b)}(\theta)
\quad
\text{s.t.}
\quad
h(\theta)\le \hat c^{(b)},
\]
where \(\phi_n^{(b)}(\theta)\) denotes the bootstrap loss and \(h(\theta)=g(\theta)^{\prime}\Sigma^{-1}g(\theta)\). Thus, each bootstrap replication incorporates both inner-loop estimation uncertainty and outer-loop selection uncertainty.

The ISP in replication \(b\) is evaluated at the bootstrap constrained estimator, namely
\[
\widehat{\mathrm{ISP}}^{(b)}
=
2\hat{\lambda}^{(b)}(\hat c^{(b)})\,
\mathrm{diag}\!\Big(\mathrm{sign}\big(g(\hat{\theta}^{(b)}(\hat c^{(b)}))\big)\Big)\,
\Sigma^{-1}g(\hat{\theta}^{(b)}(\hat c^{(b)})).
\]
We emphasize that ISP is evaluated at the constrained estimator within the inner problem, rather than at the debiased estimator, since otherwise \(g(\theta)\) may be driven artificially close to zero and remove the variation needed for relevance assessment.

\begin{assumption}[Bootstrap score validity]
\label{ass:wildB}
Let the bootstrap score be
\[
\psi_n^{\mathcal B}(\theta;c)=n^{-1}\sum_{i=1}^n \psi_i^{\mathcal B}(\theta;c).
\]
For each fixed \(c\in\mathcal C_{\mathrm{active}}\), conditionally in probability,
\begin{equation*}
\sqrt n\Big(\psi_n^{\mathcal B}(\hat\theta(c);c)-\psi_n(\hat\theta(c))\Big)
=
\frac{1}{\sqrt n}\sum_{i=1}^n \psi_i(\hat\theta(c))\,w_i
+
o_p^{\mathcal B}(1).
\end{equation*}
\end{assumption}

Assumption~\ref{ass:wildB} is standard for bootstrap validity of smooth \(M\)-estimators. Given the KKT-based asymptotic linear representation, bootstrap validity reduces to verifying that the bootstrap score reproduces the same first-order linear structure, which we impose as a high-level condition. In the optimized procedure, bootstrap samples are centered at the estimator evaluated at the selected tolerance \(\hat c\), and the full estimation routine is repeated in each replication.

\begin{theorem}[Wild bootstrap consistency for \((\hat\theta(c),\hat\lambda(c))\) at fixed \(c\)]
\label{thm:wild}
Under Assumptions~\ref{assump:slater}, \ref{assump:smooth_cvx}, \ref{assump:regular}, \ref{assump:S_CLT}, and \ref{ass:wildB}, for each fixed \(c\in\mathcal C_{\mathrm{active}}\), conditional on the data,
\begin{equation*}
\sqrt n
\binom{\hat\theta^{\mathcal B}(c)-\hat\theta(c)}{\hat\lambda^{\mathcal B}(c)-\hat\lambda(c)}
\ \xrightarrow{d^\mathcal B}\
\sqrt n
\begin{pmatrix}
\hat\theta(c)-\theta^*(c)\\
\hat\lambda(c)-\lambda^*(c)
\end{pmatrix}.
\end{equation*}
\end{theorem}

\noindent \textbf{Proof.} See Supplement~\ref{sec:pf_wild-bootstrap}.

\begin{theorem}[Wild bootstrap consistency for ISP at fixed \(c\)]
\label{thm:wildISP}
Under the same conditions, for each fixed \(c\in\mathcal C_{\mathrm{active}}\), conditional on the data,
\begin{equation*}
\sqrt n\big(\widehat{ISP}^{\mathcal B}(c)-\widehat{ISP}(c)\big)
\xrightarrow{d^{\mathcal B}}
N(0,\Sigma_{ISP}(c)).
\end{equation*}
\end{theorem}

\noindent \textbf{Proof.} See Supplement~\ref{sec:pf_wildISP}.

Theorems~\ref{thm:wild}--\ref{thm:wildISP} justify the bootstrap for the smooth inner-level objects at a fixed tolerance \(c\). In practice, however, we implement the bootstrap after the optimal tolerance \(\hat c\) has been selected and then re-run the full outer-inner estimation procedure within each bootstrap replication. This yields bootstrap draws \((\hat c^{(b)},\hat\theta^{(b)}(\hat c^{(b)}),\hat\lambda^{(b)}(\hat c^{(b)}),\widehat{ISP}^{(b)})\), which incorporate both selection uncertainty and estimation uncertainty. Accordingly, the fixed-\(c\) bootstrap theory serves as the inner-level justification, while the practical algorithm propagates this uncertainty through the optimized choice of \(\hat c\).
\subsection{Bootstrap-based plateau rule}

As discussed above, bootstrap implementation of the plateau rule is particularly important because the separation condition in Assumption~\ref{assump:ISP_separation} is not directly verifiable in practice and may be weak when several restrictions have similar empirical effects. In such cases, small perturbations in the data can alter the ordering of ISP magnitudes and shift the location of the cutoff. The bootstrap addresses this issue by repeatedly perturbing the sample and recomputing the entire plateau-selection procedure, thereby revealing the sensitivity of the cutoff to sampling variability. To this end, we apply the plateau rule within each bootstrap replication and use the resulting distribution to characterize cutoff uncertainty.

For each bootstrap replication \(b=1,\dots,B\), let
$
|\widehat{ISP}^{(b)}_{(1)}(c)| \le \cdots \le |\widehat{ISP}^{(b)}_{(q)}(c)|
$
denote the ordered bootstrap ISP magnitudes. Repeating the same structural break and homogeneity-screening procedure yields a bootstrap cutoff
$
\hat m^{(b)}(c).
$
The collection
$
\bigl\{\hat m^{(b)}(c)\bigr\}_{b=1}^B
$
thus defines an empirical distribution of plateau cutoffs.

This bootstrap distribution serves as the primary object of inference for plateau determination. It provides an operational measure of empirical separability in the ISP sequence: when ISP magnitudes are well separated, the bootstrap cutoffs concentrate around a stable value; when separation is weak, the distribution of \(\hat m^{(b)}(c)\) becomes diffuse, indicating ambiguity in the plateau boundary. In this way, the bootstrap translates an otherwise untestable separation condition into an observable pattern of variability.

In practice, we summarize the bootstrap distribution by its mean and standard deviation,
\[
\bar m_B(c)=\frac{1}{B}\sum_{b=1}^B \hat m^{(b)}(c),
\qquad
s_B(c)=\left(\frac{1}{B-1}\sum_{b=1}^B\big(\hat m^{(b)}(c)-\bar m_B(c)\big)^2\right)^{1/2},
\]
and use these summaries to construct an uncertainty-adjusted cutoff. In particular, we define the uncertainty-adjusted plateau boundary as \(\bar m_B(c) + s_B(c)\), which expands the plateau when cutoff selection is unstable across bootstrap replications. This provides a conservative and data-driven rule for separating noise from signal.

Accordingly, our practical inference is based on the distribution of \(\hat m^{(b)}(c)\) rather than solely on the point estimate \(\hat m(c)\). While the consistency result for \(\hat m(c)\) remains a useful theoretical benchmark under Assumption~\ref{ass:plateau_sep}, the bootstrap cutoff distribution offers a more relevant and robust tool for applied implementation.

\section{Monte Carlo simulation}
In this section, we study the finite-sample performance of our proposed method through Monte Carlo simulations. We set the number of iterations to 1{,}000 and the sample size to $n = 1{,}000$. The parameter vector has dimension $p+1=11$, consisting of an intercept and $p=10$ slope coefficients.\footnote{The intercept is included to capture the baseline level of the outcome variable but is excluded from the shrinkage restriction, as penalizing it lacks substantive interpretation and would shift the model’s location rather than regulate structural relationships among slope parameters.} 

\paragraph{DGP.} Let $X = [\mathbf{1}, X_{\text{raw}}] \in \mathbb{R}^{n \times (p+1)}$ denote the regressor matrix, where $\mathbf{1}$ is an $n$-vector of ones and $X_{\text{raw}} = Z \Sigma_X^{1/2} \in \mathbb{R}^{n \times p}$ contains the slope covariates. Here, $Z \in \mathbb{R}^{n \times p}$ has i.i.d.\ standard normal entries, $Z_{ij} \overset{\text{i.i.d.}}{\sim} \mathcal{N}(0,1)$, and $\Sigma_X \in \mathbb{R}^{p \times p}$ follows an AR(1) structure, $(\Sigma_X)_{jl} = \rho^{|j-l|}$ for $\rho \in (0,1)$, inducing cross-covariate dependence. Throughout, we set $\rho = 0.8$ to generate strong correlation among the $p=10$ slope variables.

The outcome is generated as $y = X \theta^{0} + \varepsilon$, where $\theta^{0} \in \mathbb{R}^{p+1}$ is the true parameter vector and $\varepsilon \sim \mathcal{N}(0, \sigma_\varepsilon^{2} I_n)$. The disturbance variance is calibrated to achieve a target signal-to-noise ratio of approximately one, $\mathrm{Var}(X\theta^{0}) / \mathrm{Var}(\varepsilon) \approx 1$, corresponding to $R^{2} \approx 0.5$, so that the signal is dominant but the noise remains non-negligible.

\paragraph{Scenarios and parameters.}
We consider the least squares loss $\phi_n(\theta) = \frac{1}{2n}\,\lVert y - X\theta \rVert^2$ and study three cases. In Case 1, only structural restrictions are imposed. In Cases 2 and 3, we additionally impose the restriction $\theta_1 + \theta_2 + \theta_3 + \theta_4 = 0$, which is correctly specified in Case 2 and misspecified in Case 3.

The parameter vector $\theta \in \mathbb{R}^{11}$ consists of an intercept and ten slope coefficients. In Case 1,
$
\theta^{0}_{\text{Case 1}} = (0.1, 0.3, 0, -0.5, 0, 0.3, 0, 0.4, 0, -0.2, 0)^{\prime}.
$
In Case 2,
$
\theta^{0}_{\text{Case 2}} = (0.1, 0.3, 0.2, -0.5, 0, 0.3, 0, 0.4, 0, -0.2, 0)^{\prime},
$
so that the imposed restriction holds at the population level, as
$
\theta_1^{0} + \theta_2^{0} + \theta_3^{0} + \theta_4^{0}
= 0.3 + 0.2 - 0.5 + 0.0 = 0.
$
In Case 3,
$
\theta^{0}_{\text{Case 3}} = (0.1, 0.3, 0, -0.5, 0, 0.3, 0, 0.4, 0, -0.2, 0)^{\prime},
$
for which the restriction is violated, since
$\theta_1^{0} + \theta_2^{0} + \theta_3^{0} + \theta_4^{0}
= 0.3 + 0.0 - 0.5 + 0.0 = -0.2 \neq 0.
$

We adopt the credibility matrix $\Sigma = I$, assigning equal weight to all restrictions and reflecting the absence of prior preference or differential confidence, thereby capturing model uncertainty. This specification, combined with restrictions imposed only on the structural parameters, coincides with the standard ridge formulation. The tolerance parameter is constrained to $c \in (0, c_0]$ with $c_0 = 1$, and is selected by minimizing the Stein-type risk criterion.

\paragraph{Results.}
The results show several desirable properties of our proposed framework. First, the estimator consistently recovers the true data-generating parameters under model uncertainty, regardless of whether the imposed restrictions hold, as shown in Table~\ref{tab:estimates}. This reflects that the regularization bias induced by potentially misspecified restrictions is effectively controlled through the optimized tolerance and subsequent debiasing step. Second, inference achieves the intended nominal coverage, with empirical coverage rates close to the 95\% level and comparable to those obtained from bootstrap-based procedures (Table~\ref{tab:nominal-coverage-stat}). Third, the plateau is clearly identified by the shaded region in Figure~\ref{fig:plateau}, which corresponds to the uncertainty-adjusted cutoff. The restrictions classified within this plateau are marked with * in Table~\ref{tab:estimates}. Notably, these restrictions are associated with true parameter values equal to zero and exhibit little empirical relevance for model fit. When an additional restriction holds at the population level (Case 2), it is correctly included in the plateau. In contrast, when the restriction is violated at the population level (Case 3), it falls outside the plateau, indicating that relaxing the restriction improves model fit. Thus, the plateau rule effectively distinguishes between empirically irrelevant and relevant restrictions.

\begin{table}[!htbp] \centering
\caption{Estimation results} \label{tab:estimates}
\resizebox{0.93\textwidth}{!}{
\begin{tabular}{rcccccc} \hline
\multicolumn{1}{l}{}             & \multicolumn{2}{c}{Case 1}         & \multicolumn{2}{c}{Case 2}          & \multicolumn{2}{c}{Case 3}           \\ \hline
Parameters $(j=1,\dots,10)$      & \multicolumn{6}{c}{$\theta_{1}=0,\theta_{2}=0,\dots,\theta_{10}=0$}                                             \\
Additional $(j=11)$              & \multicolumn{6}{c}{$\theta_{1}+\theta_{2}+\theta_{3}+\theta_{4}=0$}                                             \\ \cline{2-7}
\multicolumn{1}{l}{}             & \multicolumn{2}{c}{$\textit{N/A}$} & \multicolumn{2}{c}{$\textit{True}$} & \multicolumn{2}{c}{$\textit{False}$} \\ \hline
$\hat{\lambda}$                  & \multicolumn{2}{c}{0.0034}   & \multicolumn{2}{c}{0.0037}   & \multicolumn{2}{c}{0.0033}   \\
\multicolumn{1}{l}{}             & \multicolumn{2}{c}{(0.0009)} & \multicolumn{2}{c}{(0.0010)} & \multicolumn{2}{c}{(0.0009)} \\ \hline
True parameters                  & Estimates      & $ISP_{j}$   & Estimates      & $ISP_{j}$   & Estimates      & $ISP_{j}$   \\ \hline
$\theta^{0}_{0}$                 & 0.1002         &             & 0.1002          &            & 0.1002          &            \\
                                 & (0.0163)       &             & (0.0185)        &            & (0.0163)        &            \\
                                 & {[}0.0002{]}   &             & {[}0.0002{]}    &            & {[}0.0002{]}    &            \\
$\theta^{0}_{1}$                 & 0.3018         & 0.0021      & 0.3020          & 0.0022     & 0.3018          & 0.0020     \\
                                 & (0.0283)       & (0.0006)    & (0.0321)        & (0.0007)   & (0.0283)        & (0.0006)   \\
                                 & {[}0.0018{]}   &             & {[}0.0020{]}     &            & {[}0.0018{]}    &            \\
$\theta^{0}_{2}$                 & -0.0030        & 0.0002      & 0.1967          & 0.0014     & -0.0030         & 0.0002     \\
                                 & (0.0356)       & (0.0002)    & (0.0404)        & (0.0004)   & (0.0356)        & (0.0002)   \\
                                 & {[}0.0030{]}    & *           & {[}0.0033{]}    &            & {[}0.0030{]}     & *          \\
$\theta^{0}_{3}$                 & -0.4987        & 0.0034      & -0.4985         & 0.0036     & -0.4987         & 0.0032     \\
                                 & (0.0346)       & (0.0008)    & (0.0393)        & (0.0008)   & (0.0347)        & (0.0008)   \\
                                 & {[}0.0013{]}   &             & {[}0.0015{]}    &            & {[}0.0013{]}    &            \\
$\theta^{0}_{4}$                 & 0.0005         & 0.0002      & 0.0005          & 0.0002     & 0.0004          & 0.0002     \\
                                 & (0.0338)       & (0.0001)    & (0.0384)        & (0.0002)   & (0.0338)        & (0.0001)   \\
                                 & {[}0.0005{]}   & *           & {[}0.0005{]}    & *          & {[}0.0004{]}    & *          \\
$\theta^{0}_{5}$                 & 0.2992         & 0.0020      & 0.2991          & 0.0021     & 0.2992          & 0.0019     \\
                                 & (0.0350)       & (0.0005)    & (0.0397)        & (0.0006)   & (0.0350)        & (0.0005)   \\
                                 & {[}0.0008{]}   &             & {[}0.0009{]}    &            & {[}0.0008{]}    &            \\
$\theta^{0}_{6}$                 & 0.0005         & 0.0002      & 0.0005          & 0.0002     & 0.0005          & 0.0002     \\
                                 & (0.0354)       & (0.0002)    & (0.0402)        & (0.0002)   & (0.0354)        & (0.0002)   \\
                                 & {[}0.0005{]}   & *           & {[}0.0005{]}    & *          & {[}0.0005{]}    & *          \\
$\theta^{0}_{7}$                 & 0.4008         & 0.0027      & 0.4009          & 0.0029     & 0.4008          & 0.0026     \\
                                 & (0.0365)       & (0.0006)    & (0.0414)        & (0.0007)   & (0.0365)        & (0.0006)   \\
                                 & {[}0.0008{]}   &             & {[}0.0009{]}    &            & {[}0.0008{]}    &            \\
$\theta^{0}_{8}$                 & -0.0009        & 0.0002      & -0.0010         & 0.0002     & -0.0009         & 0.0002     \\
                                 & (0.0353)       & (0.0002)    & (0.0401)        & (0.0002)   & (0.0353)        & (0.0002)   \\
                                 & {[}0.0009{]}   & *           & {[}0.0010{]}     & *          & {[}0.0009{]}    & *          \\
$\theta^{0}_{9}$                 & -0.2010        & 0.0014      & -0.2011         & 0.0014     & -0.2010         & 0.0013     \\
                                 & (0.0348)       & (0.0004)    & (0.0395)        & (0.0004)   & (0.0348)        & (0.0004)   \\
                                 & {[}0.0010{]}    &             & {[}0.0011{]}    &            & {[}0.0010{]}     &            \\
$\theta^{0}_{10}$                & 0.0006         & 0.0002      & 0.0006          & 0.0002     & 0.0006          & 0.0001     \\
                                 & (0.0275)       & (0.0001)    & (0.0312)        & (0.0001)   & (0.0275)        & (0.0001)   \\
                                 & {[}0.0006{]}   & *           & {[}0.0006{]}    & *          & {[}0.0006{]}    & *          \\
$\theta^{0}_{1}+\theta^{0}_{2}$  & \multicolumn{2}{c}{}         &                 & 0.0002     &                 & 0.0013     \\
$+\theta^{0}_{3}+\theta^{0}_{4}$ &                & \textit{}   &                 & (0.0002)   &                 & (0.0003)   \\
                                 &                & \textit{}   &                 & *          &                 &            \\ \hline\hline
$R^{2}$                          & \multicolumn{2}{c}{0.5057}   & \multicolumn{2}{c}{0.5055}   & \multicolumn{2}{c}{0.5057}   \\ \hline
$\hat{c}$                        & \multicolumn{2}{c}{0.6230}   & \multicolumn{2}{c}{0.6621}   & \multicolumn{2}{c}{0.6620}   \\
                                 & \multicolumn{2}{c}{(0.0545)} & \multicolumn{2}{c}{(0.0655)} & \multicolumn{2}{c}{(0.0579)} \\
Risk estimate  & \multicolumn{2}{c}{0.0093}   & \multicolumn{2}{c}{0.0118}   & \multicolumn{2}{c}{0.0093}   \\
                 & \multicolumn{2}{c}{(0.0006)} & \multicolumn{2}{c}{(0.0008)} & \multicolumn{2}{c}{(0.0006)} \\
Bias proxy       & \multicolumn{2}{c}{0.0002}   & \multicolumn{2}{c}{0.0003}   & \multicolumn{2}{c}{0.0002}   \\
                 & \multicolumn{2}{c}{(0.0001)} & \multicolumn{2}{c}{(0.0001)} & \multicolumn{2}{c}{(0.0001)} \\
Variance proxy & \multicolumn{2}{c}{0.0091}   & \multicolumn{2}{c}{0.0115}   & \multicolumn{2}{c}{0.0091}   \\
                 & \multicolumn{2}{c}{(0.0006)} & \multicolumn{2}{c}{(0.0008)} & \multicolumn{2}{c}{(0.0006)} \\ \hline
\multicolumn{7}{l}{%
\begin{minipage}{17cm}\vspace{1mm}
    \footnotesize Note: Results are based on 1,000 Monte Carlo iterations. Standard deviations are reported in parentheses, and average absolute biases are shown in brackets. An asterisk (*) indicates restrictions classified within the plateau—i.e., those falling below the uncertainty-adjusted cutoff (mean + sd)—thereby accounting for sampling variability in the estimated cutoff.
\end{minipage}}\\
\end{tabular}}
\end{table} 

\begin{table}[!htbp]
\centering
\caption{Summary statistics for nominal coverage across cases}
\label{tab:nominal-coverage-stat}
\begin{tabular}{rcccccc} \hline
\multicolumn{1}{l}{} & \multicolumn{2}{c}{Case 1}         & \multicolumn{2}{c}{Case 2}          & \multicolumn{2}{c}{Case 3}           \\ \hline
Parameters $(j=1,\dots,10)$                          & \multicolumn{6}{c}{$\theta_{1}=0,\theta_{2}=0,\dots,\theta_{10}=0$}                                             \\
Additional $(j=11)$                                 & \multicolumn{6}{c}{$\theta_{1}+\theta_{2}+\theta_{3}+\theta_{4}=0$}                                             \\ \cline{2-7}
                                                    & \multicolumn{2}{c}{$\textit{N/A}$} & \multicolumn{2}{c}{$\textit{True}$} & \multicolumn{2}{c}{$\textit{False}$} \\ \cline{2-7}
\multicolumn{1}{c}{} & Analyt. & Boot.   & Analyt. & Boot.   & Analyt. & Boot.   \\ \hline
Minimum & 92.50\%  & 93.30\% & 92.50\% & 93.30\% & 92.50\% & 93.30\% \\
Q1      & 93.20\%  & 93.60\% & 93.15\% & 93.60\% & 93.40\% & 93.60\% \\
Median  & 93.60\% & 94.00\% & 93.60\% & 94.00\% & 93.50\% & 94.00\% \\
Mean    & 93.74\%  & 94.05\% & 93.77\% & 94.06\% & 93.78\% & 94.05\% \\
Q3      & 94.50\%  & 94.55\% & 94.50\% & 94.55\% & 94.45\% & 94.55\% \\
Maximum & 94.90\%  & 94.80\% & 94.90\% & 94.80\% & 94.90\% & 94.80\% \\ \hline
\multicolumn{7}{l}{%
\begin{minipage}{15cm}\vspace{1mm}
    \scriptsize Note: The results are based on 1,000 iterations, with the nominal confidence level set at 95\%. \textit{Analyt.} refers to nominal coverage using the analytic variance formulas derived in Theorem~\ref{thm:asymp_theta_active}, and \textit{Boot.} denotes nominal coverage based on the wild bootstrap inference established in Theorem~\ref{thm:wild}. Case 1 corresponds to ridge regression as a special case of our framework, in which restrictions are only imposed on the parameters. Case 2 represents our method with an additional restriction that is correctly specified and holds in the population. Case 3 represents our method with the same additional restriction imposed, but the restriction is misspecified and does not hold in the population. Q1 and Q3 denote the first and third quantiles, respectively.
\end{minipage}}\\
\end{tabular}
\end{table}

\begin{figure}
    \caption{Sorted ISP values (ascending) with plateau region up to the cutoff}
    \label{fig:plateau}
    \centering
    \includegraphics[width=\linewidth]{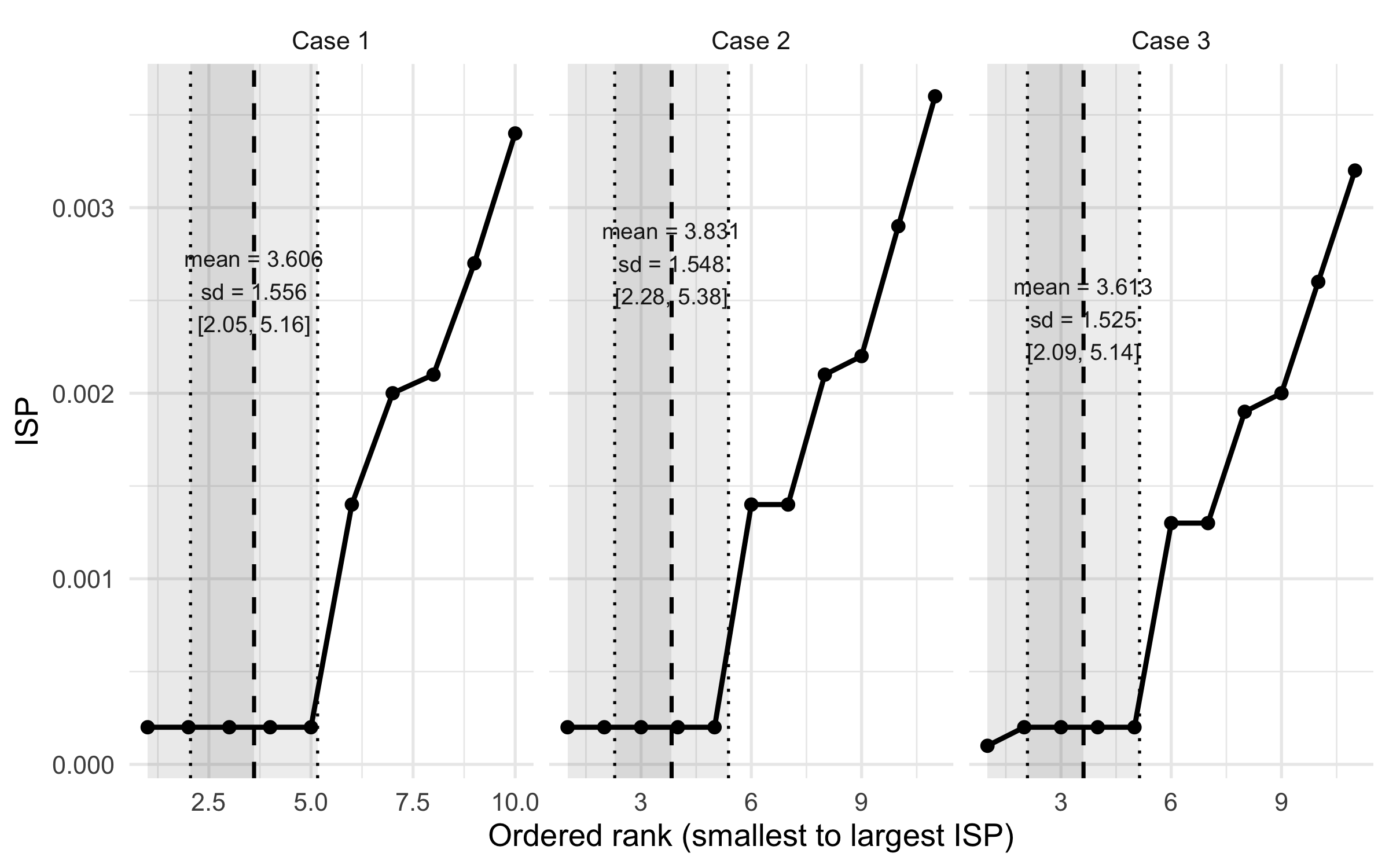}
    \begin{minipage}{15cm}\vspace{1mm}
    \scriptsize Note: The results are based on 1,000 iterations. ISPs are sorted in ascending order. Bold dashed: mean cutoff; dotted: mean cutoff ± standard deviation; shaded: uncertainty-adjusted plateau.
\end{minipage}\\
\end{figure}

\section{Empirical Application}

In this section, we apply our proposed estimation method to the textbook Solow growth model \citep{Solow1956, Swan1956}. The canonical empirical specification includes the logarithm of the saving rate and the logarithm of the effective population growth rate, capturing the accumulation and dilution channels of capital in steady state.

The Solow model implies a cross-parameter restriction that the coefficients on saving and effective population growth are equal in magnitude but opposite in sign. Specifically,
\begin{equation}
\ln y_i
=
c
+ \frac{\alpha}{1-\alpha}\ln s_i
- \frac{\alpha}{1-\alpha}\ln(n_i+g+\delta)
+ \varepsilon_i,
\end{equation}
which yields the restriction
$
\theta_s + \theta_n = 0.
$ While this restriction captures a first-order implication of the theory, it relies on strong assumptions such as homogeneity and common adjustment dynamics. In practice, growth processes may exhibit nonlinearities and heterogeneous dynamics \citep{DurlaufandJohnson1995, Temple1999, Quah1996, DurlaufandQuah1999}. To account for such model uncertainty, we consider structured nonlinear relaxations of the Solow relation:
\begin{equation}
g(\theta) \equiv \theta_s - (-\theta_n)^{\tau} = 0, \quad \tau=2,3,
\end{equation}
which allow for curvature while preserving qualitative predictions. The inclusion of polynomial terms—particularly up to the cubic case—is intended for illustrative purposes, providing a simple and tractable way to capture potential nonlinearities rather than a comprehensive specification of the underlying growth process.

In specifying the set of candidate restrictions, we focus on economically meaningful hypotheses. In particular, we do not impose restrictions such as $\theta_s=0$ or $\theta_n=0$, which are inconsistent with the qualitative implications of the Solow model. Instead, we evaluate the empirical relevance of the linear Solow restriction and its nonlinear relaxations. To reflect equal prior credibility across these restrictions, we set $\Sigma = I_3$.

The estimation results are reported in Table \ref{tab:empirics_results}. Our approach imposes \emph{soft} rather than exact restrictions, allowing the data to determine the extent to which each restriction is enforced. The selected tolerance parameter is $\hat{c} = 379.7468$, with a Stein-type risk estimate of $0.0711$, decomposed into a negligible bias proxy ($0.0003$) and a dominant variance component ($0.0708$). Consistent with this, the estimated coefficients coincide with the unrestricted benchmark, and the $R^2$ remains unchanged at $0.6272$, indicating no gain from relaxing the restriction.

At the same time, the Wald test rejects the null hypothesis $\theta_s + \theta_n = 0$ (p-value $=0.0380$), indicating that the restriction is statistically invalid. However, the ISP-based plateau rule classifies it as empirically irrelevant. This illustrates a key distinction: statistical validity concerns whether a restriction holds exactly, whereas empirical relevance concerns whether its violation matters for estimation. In this application, deviations from the Solow restriction occur in a direction that is effectively flat for the objective function, so relaxing it does not improve model fit.

In contrast, the cubic restriction $\theta_s - (-\theta_n)^3 = 0$ lies outside the plateau, indicating that it is relatively more informative than the other restrictions. This suggests that nonlinear features capture aspects of the data that are not reflected in the linear Solow restriction. 

Notably, these findings highlight that empirical relevance is restriction-specific. A theoretically motivated restriction may be statistically invalid yet practically inconsequential, while alternative structured relaxations can provide relatively more informative features. Our method adapts accordingly, imposing restrictions only when they improve the bias--variance tradeoff and otherwise defaulting to the unrestricted benchmark without introducing unnecessary distortion.

\begin{table}
\centering
\caption{Results for the textbook and cubic specifications of the Solow model} \label{tab:empirics_results}
\begin{tabular}{rccc} \hline
\multicolumn{1}{l}{}             & \multicolumn{2}{c}{Textbook Solow}                      & \multirow{1}{*}{Our estimation} \\ \cline{2-3}
                                 & Unrestricted               & Restricted                 & (Debiased)                                \\ \hline
constant                         & 4.651          & 8.3748        & 4.651                \\
                                 & (0.0010)       & (0.0000)      & (1.4462)              \\
$\ln s$                          & 1.2756         & 1.3791        & 1.2756                \\
                                 & (0.0000)       & (0.0000)      & (0.1064)              \\
$\ln(n+0.05)$                    & -2.7087        & -1.3791       & -2.7088               \\
                                 & (0.0000)       & (0.0000)      & (0.5726)              \\ \hline \hline
Restrictions                     & \multicolumn{2}{c}{Wald}       & Plateau           \\ \hline
$\theta_{s}+\theta_{n}=0$        & \multicolumn{2}{c}{4.4273}      &  *                     \\
                                 & \multicolumn{2}{c}{(0.0380)}   &                       \\
$\theta_{s}-(-\theta_{n})^{2}=0$ &                &               & *                      \\
$\theta_{s}-(-\theta_{n})^{3}=0$ &                &               &                      \\ \hline \hline
$R^2$                            & 0.6272         & 0.6084        & 0.6272                \\ \hline
$\hat{c}$ ($c_{0}=500$)          & \multicolumn{3}{c}{379.7468}                                \\
Risk estimate                    & \multicolumn{3}{c}{0.0711}                             \\
Bias proxy                       & \multicolumn{3}{c}{0.0003}                             \\
Variance proxy                   & \multicolumn{3}{c}{0.0708}   \\ \hline       
\multicolumn{4}{l}{%
\begin{minipage}{12cm}\vspace{1mm}
    \scriptsize Note: Standard errors are reported in parentheses. $\ln s$ denotes the logarithm of the savings rate, and $\ln(n+g+\delta)$ denotes the logarithm of population growth augmented by exogenous technological progress and depreciation, which is set to 0.05. In the \textit{Textbook Solow} column, “Unrestricted” refers to estimation without imposing the theoretical Solow restriction, while “Restricted” imposes the steady-state condition that the coefficients on $\ln s$ and $\ln(n+g+\delta)$ sum to zero. The reported $R^{2}$ measures the proportion of variation in the dependent variable explained by the model. The risk estimate refers to the analytical Stein-type risk proxy used to select the tolerance parameter, which decomposes into a bias proxy capturing regularization-induced distortion and a variance proxy reflecting estimation uncertainty.
\end{minipage}}\\
\end{tabular}
\end{table}

\section{Conclusion}
This paper develops a framework for estimation and empirical relevance inference under model uncertainty with multiple restrictions. By combining Lagrangian-constrained optimization with a data-driven tolerance parameter, the estimator adapts to the bias--variance tradeoff, while a debiasing step restores valid inference. The proposed individual shadow prices (ISP) and plateau rule provide a practical way to assess which restrictions matter for improving model fit.

The framework can be naturally interpreted within a semiparametric perspective. The object of interest is a finite-dimensional target component, while potentially complex nuisance components are treated as given. This structure is compatible with modern machine learning methods, which can be used to flexibly estimate nuisance components without affecting inference on the target parameter.

Our theoretical results establish consistency and asymptotic normality, and the empirical application illustrates how the method distinguishes between statistical validity and empirical relevance. In particular, restrictions that are rejected by standard tests need not be useful for improving estimation, while alternative structured relaxations can capture relatively more informative features.

An important direction for future research is to extend this framework to high-dimensional settings by integrating it with double/debiased machine learning (DML) \citep{Chernozhukovetal2018}. In such settings, nuisance components can be flexibly estimated using machine learning methods, while orthogonalization and cross-fitting ensure that inference on the finite-dimensional target remains valid. Embedding our approach within a DML framework would allow theory-driven restrictions to be imposed on the target component while accommodating high-dimensional controls, thereby preserving valid inference under model uncertainty. More broadly, this integration would enable scalable evaluation of a large set of candidate restrictions in complex, high-dimensional environments.

\subsection*{Acknowledgement}
We thank Zheng Fang, Florian Gunsilius (Emory University), Nikolay Gospodinov (Federal Reserve Bank of Atlanta), David Drukker (Clemson University), and the audience at the Georgia Econometrics Workshop for helpful comments. All remaining errors are our own.

\clearpage 
\bibliographystyle{apalike} 
\bibliography{main} 

\newpage 
\appendix 
\setcounter{equation}{0} \renewcommand{\theequation}{A.\arabic{equation}} \renewcommand{\thesection}{A} \renewcommand{\thesubsection}{A.\arabic{subsection}}

\section*{Appendix}

\subsection{Proof for Lemma \ref{lem:population_profile}}
\label{subsec:pf_lem-population-profile}

Consider the map
$
(\theta,\lambda,c)\mapsto \Phi(\theta,\lambda;c)
$
from \(\mathbb R^p\times\mathbb R\times\mathbb R\) into \(\mathbb R^{p+1}\).
By Assumption~\ref{assump:profile_regularity}(1), \(\Phi\) is continuously differentiable
in a neighborhood of
$
(\theta^*(c^{\dagger}),\lambda^*(c^{\dagger}),c^{\dagger}).
$
By Assumption~\ref{assump:profile_regularity}(2),
$
\Phi(\theta^*(c^{\dagger}),\lambda^*(c^{\dagger});c^{\dagger})=0.
$
By Assumption~\ref{assump:profile_regularity}(3), the Jacobian of \(\Phi\) with respect
to \((\theta,\lambda)\) at this point is nonsingular. Therefore, by the implicit
function theorem, there exists an open neighborhood \(\mathcal N(c^{\dagger})\) of
\(c^{\dagger}\) and unique continuously differentiable functions
$
c\mapsto (\theta^*(c),\lambda^*(c))
$
such that
$
\Phi(\theta^*(c),\lambda^*(c);c)=0
$
for all \(c\in\mathcal N(c^{\dagger})\). The second block of \(\Phi(\theta^*(c),\lambda^*(c);c)=0\) gives
$
h(\theta^*(c))=c.
$
Since \(\lambda^*(c^{\dagger})>0\) and \(\lambda^*(c)\) is continuous in \(c\), possibly
after shrinking \(\mathcal N(c^{\dagger})\), we obtain
$
\lambda^*(c)>0
\;\text{for all } c\in\mathcal N(c^{\dagger}).
$
\qed

\subsection{Proof for Theorem \ref{thm:sample_profile}}
\label{subsec:pf_thm-sample-profile}

By Lemma~\ref{lem:population_profile}, there exists a neighborhood
\(\mathcal N(c^{\dagger})\) on which the population profiling map
$
c\mapsto (\theta^*(c),\lambda^*(c))
$
is uniquely defined and continuously differentiable, with \(\lambda^*(c)>0\). By Assumption~\ref{assump:sample_profile}(3), the sample KKT map \(\Phi_n\) is
uniformly close to the population KKT map \(\Phi\) on a neighborhood of the graph of
this solution map. By Assumption~\ref{assump:sample_profile}(2), the Jacobian of
\(\Phi_n\) with respect to \((\theta,\lambda)\) is nonsingular in that neighborhood
with probability approaching one. Hence, applying the implicit function theorem to
\(\Phi_n\) pointwise in \(c\), we obtain that, with probability approaching one, for
each \(c\in\mathcal N(c^{\dagger})\), there exists a unique local solution
$
(\hat\theta(c),\hat\lambda(c))
$
to the equation
$
\Phi_n(\theta,\lambda;c)=0.
$
Because \(\Phi_n\) is continuously differentiable, this solution is continuously
differentiable in \(c\). The second block of \(\Phi_n(\hat\theta(c),\hat\lambda(c);c)=0\) implies
$
h(\hat\theta(c))=c.
$
Since \((\hat\theta(c),\hat\lambda(c))\) converges locally to
\((\theta^*(c),\lambda^*(c))\) and \(\lambda^*(c)>0\) on \(\mathcal N(c^{\dagger})\),
possibly after shrinking the neighborhood, we obtain
$
\hat\lambda(c)>0
$
with probability approaching one, uniformly for \(c\in\mathcal N(c^{\dagger})\).
\qed

\subsection{Proof for Corollary \ref{cor:profiling_valid}}
\label{subsec:pf_cor-profiling-valid}

This follows immediately from Theorem~\ref{thm:sample_profile}. The existence and
uniqueness of \((\hat\theta(c),\hat\lambda(c))\) for each \(c\) guarantee that profiling
is well defined on the active region, while continuous differentiability in \(c\) ensures that the profiled criterion
$
c\mapsto \widehat R_{\mathrm{lin}}(c)
$
is itself locally well behaved.
\qed

\subsection{Proof of Theorem \ref{thm:asymp_theta_active}}
\label{subsec:derivation_asymp}

Fix any \(c\in\mathcal C_{\mathrm{active}}\) and let
$\hat{\theta}(c)\in\arg\min_{\theta}\ \phi_n(\theta)
\quad\text{s.t.}\quad
h(\theta)\equiv g(\theta)^{\prime}\Sigma^{-1}g(\theta)\le c,
$
with sample Lagrangian
$
\mathcal{L}_n(\theta,\lambda;c)
\equiv \phi_n(\theta)+\lambda\bigl[h(\theta)-c\bigr],
\; \lambda\ge 0.
$
Let \((\theta^*(c),\lambda^*(c))\) denote the population KKT solution, and because \(c\in\mathcal C_{\mathrm{active}}\),
$
h(\theta^*(c))=c,\; \lambda^*(c)>0.
$

Define the score and Hessian objects
$
\psi_n(\theta)\equiv \nabla_\theta \phi_n(\theta)
=\frac{1}{n}\sum_{i=1}^n \psi(Z_i;\theta),
\;
H(\theta)\equiv \nabla^2_{\theta\theta}\phi(\theta),$
and set
$
a(c) \equiv \nabla_\theta h(\theta^*(c)),
\;
A(c) \equiv \nabla^2_{\theta\theta}\mathcal{L}(\theta^*(c),\lambda^*(c);c)
=
H(\theta^*(c))+\lambda^*(c) \nabla^2_{\theta\theta}h(\theta^*(c)).
$
Assume \(A(c)\) is nonsingular and \(a(c)^{\prime}A(c)^{-1}a(c)>0\).

\paragraph{Step 1: Sample and population KKT systems.}
In the binding regime, the sample KKT conditions for \((\hat{\theta}(c),\hat{\lambda}(c))\) are
\begin{align}
\nabla_\theta \mathcal{L}_n(\hat{\theta}(c),\hat{\lambda}(c);c) &= 0, \label{eq:kkt1_sample_c}\\
h(\hat{\theta}(c)) &= c, \label{eq:kkt2_sample_c}\\
\hat{\lambda}(c) &> 0. \label{eq:kkt3_sample_c}
\end{align}
The population KKT conditions for \((\theta^*(c),\lambda^*(c))\) are
\begin{align}
\nabla_\theta \mathcal{L}(\theta^*(c),\lambda^*(c);c) &= 0, \label{eq:kkt1_pop_c}\\
h(\theta^*(c)) &= c, \label{eq:kkt2_pop_c}\\
\lambda^*(c) &> 0. \label{eq:kkt3_pop_c}
\end{align}

\paragraph{Step 2: Linearization of stationarity.}
Apply a first-order Taylor expansion of \(\nabla_\theta \mathcal{L}_n(\hat{\theta}(c),\hat{\lambda}(c);c)\)
around \((\theta^*(c),\lambda^*(c))\):
\begin{equation*}\hspace{-10mm}
0
=
\nabla_\theta \mathcal{L}_n(\theta^*(c),\lambda^*(c);c)
+
\nabla^2_{\theta\theta}\mathcal{L}(\theta^*(c),\lambda^*(c);c)\big(\hat{\theta}(c)-\theta^*(c)\big)
+
\nabla^2_{\theta\lambda}\mathcal{L}(\theta^*(c),\lambda^*(c);c)\big(\hat{\lambda}(c)-\lambda^*(c)\big)
+
r_{n,1}(c),
\end{equation*}
where \(r_{n,1}(c)=o_p(n^{-1/2})\). Since
$
\nabla^2_{\theta\lambda}\mathcal{L}(\theta^*(c),\lambda^*(c);c)
=
\nabla_\theta h(\theta^*(c))=a(c),
$
and \(\nabla^2_{\theta\theta}\mathcal{L}(\theta^*(c),\lambda^*(c);c)=A(c)\), we obtain
\begin{equation}\label{eq:lin_stationarity_c}
A(c)\big(\hat{\theta}(c)-\theta^*(c)\big) + a(c)\big(\hat{\lambda}(c)-\lambda^*(c)\big)
=
-\nabla_\theta \mathcal{L}_n(\theta^*(c),\lambda^*(c);c) + o_p(n^{-1/2}).
\end{equation}

Moreover,
$\nabla_\theta \mathcal{L}_n(\theta^*(c),\lambda^*(c);c)
=
\nabla_\theta \phi_n(\theta^*(c)) + \lambda^*(c) a(c).
$
Using the population stationarity condition \eqref{eq:kkt1_pop_c},
$
\nabla_\theta \phi(\theta^*(c)) + \lambda^*(c) a(c) = 0,
$
we can write
\begin{equation*}
\nabla_\theta \mathcal{L}_n(\theta^*(c),\lambda^*(c);c)
=
\psi_n(\theta^*(c))-\nabla_\theta \phi(\theta^*(c)).
\end{equation*}
Redefining \(\Psi_n(\theta^*(c))\) as the centered score yields
\begin{equation}\label{eq:lin_stationarity_final_c}
A(c)\big(\hat{\theta}(c)-\theta^*(c)\big) + a(c)\big(\hat{\lambda}(c)-\lambda^*(c)\big)
=
-\,\psi_n(\theta^*(c)) + o_p(n^{-1/2}).
\end{equation}

\paragraph{Step 3: Linearization of the binding constraint.}
Since \(h(\hat{\theta}(c))=c\) and \(h(\theta^*(c))=c\), a Taylor expansion gives
\begin{equation*}
0
=
\nabla_\theta h(\theta^*(c))^{\prime}\big(\hat{\theta}(c)-\theta^*(c)\big)
+
r_{n,2}(c),
\end{equation*}
where \(r_{n,2}(c)=o_p(n^{-1/2})\), implying
\begin{equation}\label{eq:lin_constraint_c}
a(c)^{\prime}\big(\hat{\theta}(c)-\theta^*(c)\big)=o_p(n^{-1/2}).
\end{equation}

\paragraph{Step 4: Solve the linearized KKT system.}
Stacking \eqref{eq:lin_stationarity_final_c} and \eqref{eq:lin_constraint_c},
\begin{equation*}
\begin{pmatrix}
A(c) & a(c)\\
a(c)^{\prime}& 0
\end{pmatrix}
\begin{pmatrix}
\hat{\theta}(c)-\theta^*(c)\\
\hat{\lambda}(c)-\lambda^*(c)
\end{pmatrix}
=
\begin{pmatrix}
-\,\psi_n(\theta^*(c))\\
0
\end{pmatrix}
+
o_p(n^{-1/2}).
\end{equation*}

Let \(\Delta_\theta(c)=\hat{\theta}(c)-\theta^*(c)\) and
\(\Delta_\lambda(c)=\hat{\lambda}(c)-\lambda^*(c)\). Then
\begin{equation*}
\Delta_\theta(c)
=
-\,A(c)^{-1}\psi_n(\theta^*(c))
- A(c)^{-1}a(c)\,\Delta_\lambda(c)
+ o_p(n^{-1/2}).
\end{equation*}
Premultiplying by \(a(c)^{\prime}\) and using \eqref{eq:lin_constraint_c},
\begin{equation*}
0
=
-\,a(c)^{\prime}A(c)^{-1}\psi_n(\theta^*(c))
- a(c)^{\prime}A(c)^{-1}a(c)\,\Delta_\lambda(c)
+ o_p(n^{-1/2}),
\end{equation*}
so that
\begin{equation}\label{eq:Delta_lambda_c}
\Delta_\lambda(c)
=
-\,\bigl(a(c)^{\prime}A(c)^{-1}a(c)\bigr)^{-1}
a(c)^{\prime}A(c)^{-1}\psi_n(\theta^*(c))
+ o_p(n^{-1/2}).
\end{equation}
Substituting into \(\Delta_\theta(c)\),
\begin{align}
\Delta_\theta(c)
&=
-\,A(c)^{-1}\psi_n(\theta^*(c))
+
A(c)^{-1}a(c)\bigl(a(c)^{\prime}A(c)^{-1}a(c)\bigr)^{-1}
a(c)^{\prime}A(c)^{-1}\psi_n(\theta^*(c))
+ o_p(n^{-1/2})
\nonumber\\
&=
-\,M(c)\,\psi_n(\theta^*(c))+o_p(n^{-1/2}),
\label{eq:Delta_theta_c}
\end{align}
where
$M(c)
\equiv
A(c)^{-1}-A(c)^{-1}a(c)\bigl(a(c)^{\prime}A(c)^{-1}a(c)\bigr)^{-1}a(c)^{\prime}A(c)^{-1}.
$

\paragraph{Step 5: Joint asymptotic distribution.}
Combining \eqref{eq:Delta_theta_c} and \eqref{eq:Delta_lambda_c},
\begin{equation}\label{eq:joint_ALR_c}
\sqrt{n}
\begin{pmatrix}
\hat{\theta}(c)-\theta^*(c)\\
\hat{\lambda}(c)-\lambda^*(c)
\end{pmatrix}
=
\begin{pmatrix}
-\,M(c)\\
-\,(a(c)^{\prime}A(c)^{-1}a(c))^{-1}a(c)^{\prime}A(c)^{-1}
\end{pmatrix}
\sqrt{n}\,\psi_n(\theta^*(c))
+ o_p(1).
\end{equation}
By the multivariate central limit theorem,
\begin{equation}\label{eq:joint_CLT_c}
\sqrt{n}
\begin{pmatrix}
\hat{\theta}(c)-\theta^*(c)\\
\hat{\lambda}(c)-\lambda^*(c)
\end{pmatrix}
\xrightarrow{d}
\mathcal{N}\!\left(
0,\begin{pmatrix}
V_1(c) & V_2(c)\\
V_2(c)^{\prime} & V_3(c)
\end{pmatrix}
\right),
\end{equation}
with
$V_1(c) = M(c)\,\Sigma_{\psi}(c)\,M(c)^{\prime},\;
V_2(c) = -\,M(c)\,\Sigma_{\psi}(c)\,A(c)^{-1}a(c)\,(a(c)^{\prime}A(c)^{-1}a(c))^{-1},\;\\
V_3(c) = (a(c)^{\prime}A(c)^{-1}a(c))^{-2}
      a(c)^{\prime}A(c)^{-1}\Sigma_{\psi}(c)A(c)^{-1}a(c).$
\qed

\subsection{Proof of Theorem \ref{thm:asymp_theta_inactive}}
\label{subsec:pf_asymp_inactive}
Define the sample Lagrangian at tolerance $c$,
$
\mathcal{L}_n(\theta,\lambda;c)
=
\phi_n(\theta)+\lambda\bigl[h(\theta)-c\bigr],
\; \lambda\ge 0.
$
Let $(\theta^*(c),\lambda^*(c))$ denote a population KKT solution satisfying
$
h(\theta^*(c^{\dagger})) < c^{\dagger},
\;
\lambda^*(c^{\dagger}) = 0.
$
Assume strict slackness: there exists $\eta>0$ such that
\begin{equation}
h(\theta^*(c^{\dagger})) \le c^{\dagger} - \eta.
\label{eq:strict_slackness_opt}
\end{equation}

By Theorem~\ref{thm:c_hat_consistency}, we have
$
\hat c \xrightarrow{p} c^{\dagger}.
$ Let $\tilde{\theta}$ denote the unconstrained M-estimator,
$
\tilde{\theta}\in\arg\min_{\theta}\ \phi_n(\theta),
\;
\nabla_\theta \phi_n(\tilde{\theta})=0.
$
Standard M-estimation arguments imply
$
\tilde{\theta}\xrightarrow{p}\theta^*(c^{\dagger}).
$ By continuity of $h(\cdot)$ and \eqref{eq:strict_slackness_opt},
$
h(\tilde{\theta}) \xrightarrow{p} h(\theta^*(c^{\dagger})) < c^{\dagger}.
$
Since $\hat c \to c^{\dagger}$, it follows that
$
\Pr\{h(\tilde{\theta}) < \hat c\} \to 1.
$ Hence, with probability approaching one, the unconstrained estimator is feasible for the constrained problem evaluated at $\hat c$. Therefore,
$
\Pr\{\hat{\theta} = \tilde{\theta}\} \to 1,
\;
\Pr\{\hat{\lambda} = 0\} \to 1,
$
where the second statement follows from complementary slackness. Thus, asymptotically,
$
\nabla_\theta \phi_n(\hat{\theta})=0.
$ A first-order Taylor expansion around $\theta^*(c^{\dagger})$ gives
\[
0
=
\nabla_\theta \phi_n(\theta^*(c^{\dagger}))
+
\nabla_{\theta\theta}^2\phi(\theta^*(c^{\dagger}))(\hat{\theta}-\theta^*(c^{\dagger}))
+
o_p(n^{-1/2}),
\]
so that
$
H(\theta^*(c^{\dagger}))(\hat{\theta}-\theta^*(c^{\dagger}))
=
-\,\psi_n(\theta^*(c^{\dagger}))+o_p(n^{-1/2}).
$ Multiplying by $\sqrt{n}$ and applying the multivariate central limit theorem yields
\[
\sqrt{n}(\hat{\theta}-\theta^*(c^{\dagger}))
=
-\,H(\theta^*(c^{\dagger}))^{-1}\sqrt{n}\,\psi_n(\theta^*(c^{\dagger}))
+o_p(1)
\xrightarrow{d}
N\!\Bigl(0,\ H(\theta^*(c^{\dagger}))^{-1}\Sigma_{\psi}H(\theta^*(c^{\dagger}))^{-1}\Bigr).
\]

Finally, since $\Pr\{\hat{\lambda}=0\}\to 1$, we obtain
$
\hat{\lambda}\xrightarrow{p}0.
$
\qed

\subsection{Proof of Theorem \ref{thm:joint_ISP}}
\label{subsec:pf_ISP-joint}


\paragraph{Step 1: Write ISP as a smooth map of \((\theta,\lambda)\).}

Under Assumption~\ref{assump:sign_map}, we have
\begin{equation*}
\min_{j \in \mathcal J} |g_j(\theta^*(c))| \ge \kappa(c) > 0.
\end{equation*}
By continuity of \(g(\cdot)\), there exists a neighborhood \(\mathcal U(c)\) of \(\theta^*(c)\) such that
\begin{equation*}
\mathrm{sign}(g_j(\theta))
=
\mathrm{sign}(g_j(\theta^*(c)))
\quad
\text{for all } \theta \in \mathcal U(c) \text{ and } j \in \mathcal J.
\end{equation*}
Since \(\hat\theta(c) \xrightarrow{p} \theta^*(c)\), it follows that \(\Pr(\hat\theta(c) \in \mathcal U(c))\to 1\). Hence, with probability approaching one, the sign operator is locally constant and equal to its population counterpart. Define the diagonal sign matrix
$
S(c) \equiv \mathrm{diag}\!\big(\mathrm{sign}(g_1(\theta^*(c))),\dots,\mathrm{sign}(g_q(\theta^*(c)))\big),
$
and let \(b(\theta)\equiv \Sigma^{-1}g(\theta)\in\mathbb{R}^q\). Then, on the high-probability event \(\{\hat\theta(c) \in \mathcal U(c)\}\),
\begin{equation}\label{eq:ISP_map_c}
ISP(\theta,\lambda;c)
\equiv
2\lambda\,S(c)\,b(\theta)
=
2\lambda\,S(c)\,\Sigma^{-1}g(\theta).
\end{equation}

Because \(g(\theta)\) is continuously differentiable by Assumption~\ref{assump:regular} and \(\Sigma\) is fixed and positive definite, the map
$
(\theta,\lambda)
\mapsto
2\lambda\,S(c)\,\Sigma^{-1}g(\theta)
$
is continuously differentiable in a neighborhood of \((\theta^*(c),\lambda^*(c))\). Thus, in that neighborhood, the ISP is a smooth function of \((\theta,\lambda)\). Under the binding regime \(\lambda^*(c)>0\), the estimator therefore satisfies
$
\widehat{ISP}(c)
=
ISP(\hat\theta(c),\hat\lambda(c);c)
\;\text{w.p.a.1}.
$

\paragraph{Step 2: First-order expansion (delta method).}
Let \(G(\theta)\equiv \nabla_\theta g(\theta)\in\mathbb{R}^{q\times p}\). Since \(b(\theta)=\Sigma^{-1}g(\theta)\),
$
\nabla_\theta b(\theta)=\Sigma^{-1}G(\theta).
$
Differentiate \eqref{eq:ISP_map_c} at \((\theta^*(c),\lambda^*(c))\) to obtain the Jacobian blocks:
\begin{align}
J_\theta(c)
&\equiv
\left.\nabla_\theta ISP(\theta,\lambda;c)\right|_{(\theta^*(c),\lambda^*(c))}
=
2\lambda^*(c)\,S(c)\,\Sigma^{-1}G(\theta^*(c))
\in\mathbb{R}^{q\times p},
\label{eq:J_theta_c}\\
J_\lambda(c)
&\equiv
\left.\nabla_\lambda ISP(\theta,\lambda;c)\right|_{(\theta^*(c),\lambda^*(c))}
=
2\,S(c)\,\Sigma^{-1}g(\theta^*(c))
\in\mathbb{R}^{q\times 1}.
\label{eq:J_lambda_c}
\end{align}
Hence the first-order expansion is
\begin{equation}\label{eq:ISP_expansion_c}
\sqrt{n}\big(\widehat{ISP}(c)-ISP^*(c)\big)
=
J_\theta(c)\,\sqrt{n}(\hat\theta(c)-\theta^*(c))
+
J_\lambda(c)\,\sqrt{n}(\hat\lambda(c)-\lambda^*(c))
+o_p(1).
\end{equation}

\paragraph{Step 3: Joint asymptotic normality.}
Let
\begin{equation*}
\sqrt{n}
\begin{pmatrix}
\hat\theta(c)-\theta^*(c)\\
\hat\lambda(c)-\lambda^*(c)
\end{pmatrix}
\xrightarrow{d}
\mathcal{N}\!\left(0,\
V(c)\right),
\qquad
V(c)=
\begin{pmatrix}
V_1(c) & V_2(c)\\
V_2(c)^{\prime} & V_3(c)
\end{pmatrix},
\end{equation*}
where \(V_1(c),V_2(c),V_3(c)\) are the blocks previously derived for \((\hat\theta(c),\hat\lambda(c))\).
Define the full Jacobian
$
J(c) \equiv \begin{pmatrix} J_\theta(c) & J_\lambda(c) \end{pmatrix}\in\mathbb{R}^{q\times (p+1)}.
$
Then, by \eqref{eq:ISP_expansion_c} and Slutsky's theorem,
\begin{equation}\label{eq:ISP_joint_CLT_c}
\sqrt{n}\big(\widehat{ISP}(c)-ISP^*(c)\big)
\xrightarrow{d}
\mathcal{N}\!\big(0,\ \Sigma_{ISP}(c)\big),
\qquad
\Sigma_{ISP}(c)\equiv J(c)\,V(c)\,J(c)^{\prime}.
\end{equation}

\paragraph{Step 4: Explicit block form of \(\Sigma_{ISP}(c)\).}
Using \(V(c)=\begin{pmatrix}V_1(c) & V_2(c)\\ V_2(c)^{\prime} & V_3(c)\end{pmatrix}\), expand
\begin{equation}\label{eq:Sigma_ISP_blocks_c}
\Sigma_{ISP}(c)
=
J_\theta(c) V_1(c) J_\theta(c)^{\prime}
+
J_\theta(c) V_2(c) J_\lambda(c)^{\prime}
+
J_\lambda(c) V_2(c)^{\prime} J_\theta(c)^{\prime}
+
J_\lambda(c) V_3(c) J_\lambda(c)^{\prime}.
\end{equation}
Substituting \eqref{eq:J_theta_c}--\eqref{eq:J_lambda_c} yields an explicit expression entirely in terms of
\((\theta^*(c),\lambda^*(c))\), \(G(\theta^*(c))\), \(\Sigma\), and the joint variance blocks \((V_1(c),V_2(c),V_3(c))\).
\qed

\subsection{Proof of Lemma~\ref{lem:ISP_uniform}}
\label{subsec:pf_ISP-uniform}

Fix any \(c\in\mathcal C_{\mathrm{active}}\). By Theorem~\ref{thm:joint_ISP}, for each fixed \(j\),
$
\widehat{ISP}_j(c) - ISP_j^*(c) \xrightarrow{p} 0.
$
Fix \(\nu>0\). Since \(q\) is finite,
$
\left\{
\max_{1 \le j \le q}
\left|
\widehat{ISP}_j(c) - ISP_j^*(c)
\right| > \nu
\right\}
\subseteq
\bigcup_{j=1}^q
\left\{
\left|
\widehat{ISP}_j(c) - ISP_j^*(c)
\right| > \nu
\right\}.
$
Applying the union bound yields convergence of the maximum to 0 in probability.
\qed

\subsection{Proof of Theorem~\ref{thm:ISP_ranking_consistency}}
\label{subsec:pf_isp-rank}

Fix any \(c\in\mathcal C_{\mathrm{active}}\). Let \(\Delta(c)\) be as in Assumption~\ref{assump:ISP_separation} and define the event
\begin{equation*}
\mathcal{E}_n(c)
\equiv
\left\{
\,
\|\widehat{ISP}(c) - ISP^*(c)\|_\infty
<
\frac{\Delta(c)}{4}
\,
\right\}.
\end{equation*}

By Lemma~\ref{lem:ISP_uniform},
$
\|\widehat{ISP}(c) - ISP^*(c)\|_\infty \xrightarrow{p} 0,
$
and therefore \(\Pr(\mathcal{E}_n(c))\to 1\).

We show that on \(\mathcal{E}_n(c)\) the ordering of magnitudes is preserved.
For any index \(j\), the reverse triangle inequality implies
$
\big|
\,|\widehat{ISP}_j(c)| - |ISP_j^*(c)|
\,\big|
\le
|\widehat{ISP}_j(c) - ISP_j^*(c)|.
$
Hence on \(\mathcal{E}_n(c)\),
\begin{equation}
\label{eq:ISP_mag_bound_appendix_c}
\big|
\,|\widehat{ISP}_j(c)| - |ISP_j^*(c)|
\,\big|
<
\frac{\Delta(c)}{4}
\quad
\text{for all } j.
\end{equation}

Let \(k\in\{1,\dots,q-1\}\) and denote
\(j_k = r^*(c;k)\) and \(j_{k+1}=r^*(c;k+1)\).
Using \eqref{eq:ISP_mag_bound_appendix_c},
$
|\widehat{ISP}_{j_k}(c)|
\ge
|ISP_{j_k}^*(c)| - \frac{\Delta(c)}{4},
\;
|\widehat{ISP}_{j_{k+1}}(c)|
\le
|ISP_{j_{k+1}}^*(c)| + \frac{\Delta(c)}{4}.
$
Therefore,
\begin{align*}
|\widehat{ISP}_{j_k}(c)| - |\widehat{ISP}_{j_{k+1}}(c)|
&\ge
\left(|ISP_{j_k}^*(c)| - \frac{\Delta(c)}{4}\right)
-
\left(|ISP_{j_{k+1}}^*(c)| + \frac{\Delta(c)}{4}\right) \\
&=
\big(|ISP_{j_k}^*(c)| - |ISP_{j_{k+1}}^*(c)|\big)
-
\frac{\Delta(c)}{2} \\
&\ge
\Delta(c) - \frac{\Delta(c)}{2}
=
\frac{\Delta(c)}{2}
>
0,
\end{align*}
where the last inequality uses Assumption~\ref{assump:ISP_separation}. Thus, on \(\mathcal{E}_n(c)\),
$
|\widehat{ISP}_{r^*(c;k)}(c)|
>
|\widehat{ISP}_{r^*(c;k+1)}(c)|
\quad
\text{for all } k\in\{1,\dots,q-1\}.
$
Hence the estimated ordering \(\hat r(c)\) coincides with \(r^*(c)\) on \(\mathcal{E}_n(c)\), i.e.,
$
\mathcal{E}_n(c) \subseteq \{\hat r(c) = r^*(c)\}.
$
Taking probabilities yields
$
\Pr(\hat r(c) = r^*(c))
\;\ge\;
\Pr(\mathcal{E}_n(c))
\;\longrightarrow\; 1.
$
\qed

\subsection{Proof of Proposition \ref{prop:plateau_consistency}}
\label{subsec:pf_prop-plateau-consistency}

For each \(m\in\mathcal M\), let
$
\Delta_m(c)
=
\frac{1}{m}\sum_{\ell=1}^m |\widehat{ISP}_{(\ell)}(c)|
-
\frac{1}{q-m}\sum_{\ell=m+1}^q |\widehat{ISP}_{(\ell)}(c)|,
$
and let
$
\Delta_m^*(c)
=
\frac{1}{m}\sum_{\ell=1}^m |ISP_{(\ell)}^*(c)|
-
\frac{1}{q-m}\sum_{\ell=m+1}^q |ISP_{(\ell)}^*(c)|.
$ By Theorem~\ref{thm:joint_ISP} and finiteness of \(q\),
$
\max_{1\le j\le q}\left|\widehat{ISP}_{(j)}(c)-ISP_{(j)}^*(c)\right|\xrightarrow{p}0.
$
Therefore, for each fixed \(m\in\mathcal M\),
$
\Delta_m(c)\xrightarrow{p}\Delta_m^*(c).
$
Since \(\mathcal M\) is finite, this convergence is uniform over \(m\in\mathcal M\):
$
\max_{m\in\mathcal M}\left|\Delta_m(c)-\Delta_m^*(c)\right|\xrightarrow{p}0.
$

Because \(B_m(c)\) is a continuous function of \(\Delta_m(c)\) and \(\widehat{\mathrm{Var}}(\Delta_m(c))\), and because the corresponding variance estimator is consistent under the maintained assumptions, it follows that
$
\max_{m\in\mathcal M}\left|B_m(c)-B_m^*(c)\right|\xrightarrow{p}0,
$
where \(B_m^*(c)\) denotes the population analogue of the break statistic.

By Assumption~\ref{ass:plateau_sep}, \(B_m^*(c)\) has a unique maximizer \(m^*(c)\), separated from all other candidates by a strictly positive gap. The finite-set argmax theorem therefore implies
$
\hat m^*(c)\xrightarrow{p} m^*(c).
$
\qed

\subsection{Proof of Proposition \ref{prop:theta_star_expansion}}
\label{subsec:pf_prop-theta-star-expansion}

From the population KKT condition,
$
\nabla_\theta \phi(\theta^*(c))
+
\lambda^*(c)\nabla_\theta h(\theta^*(c))=0.
$
Apply a first-order Taylor expansion of \(\nabla_\theta \phi(\theta^*(c))\) around \(\theta^{0}\):
\[
\nabla_\theta \phi(\theta^*(c))
=
\nabla_\theta \phi(\theta^{0})
+
H_0(\theta^*(c)-\theta^{0})
+
o(\|\theta^*(c)-\theta^{0}\|).
\]
Since \(\nabla_\theta \phi(\theta^{0})=0\), this becomes
\[
H_0(\theta^*(c)-\theta^{0})
+
\lambda^*(c)\nabla_\theta h(\theta^*(c))
+
o(\|\theta^*(c)-\theta^{0}\|)=0.
\]
Using continuity of \(\nabla_\theta h(\theta)\),
$
\nabla_\theta h(\theta^*(c))=a_0+o(1).
$
Substituting yields
\[
H_0(\theta^*(c)-\theta^{0})
+
\lambda^*(c)a_0
+
o(\|\theta^*(c)-\theta^{0}\|)
+
o(\lambda^*(c))
=0.
\]
Solving for \(\theta^*(c)-\theta^{0}\) gives \eqref{eq:theta_star_expansion}.
\qed

\subsection{Proof for Theorem \ref{thm:c_hat_consistency}}
\label{subsec:pf_thm-c-hat-consistency}

By Assumption~\ref{assump:c_selection}, the feasible criterion \(\widehat R_{\mathrm{lin}}(c)\) converges uniformly in probability to the continuous population criterion \(R_{\mathrm{lin}}(c)\) on the compact set \(\mathcal C_{\mathrm{active}}\). Since \(R_{\mathrm{lin}}(c)\) has a unique minimizer \(c^{\dagger}\), the standard argmin theorem implies
$
\hat c \xrightarrow{p} c^{\dagger}.
$
\qed

\subsection{Proof for Corollary \ref{cor:theta_hat_opt_consistency}}
\label{subsec:pf_cor-theta-hat-opt-consistency}

By Theorem~\ref{thm:c_hat_consistency}, we have
$
\hat c \xrightarrow{p} c^{\dagger}.
$
From Theorem~\ref{thm:sample_profile}, the sample profiling maps
$
c \mapsto \hat{\theta}(c),
\;
c \mapsto \hat{\lambda}(c)
$
are well defined in a neighborhood of \(c^{\dagger}\) with probability approaching one, and satisfy
$
\sup_{c \in \mathcal N(c^{\dagger})}
\left\|
\hat{\theta}(c) - \theta^*(c)
\right\| \xrightarrow{p} 0,
\;
\sup_{c \in \mathcal N(c^{\dagger})}
\left|
\hat{\lambda}(c) - \lambda^*(c)
\right| \xrightarrow{p} 0.
$

We first establish consistency of \(\hat{\theta} = \hat{\theta}(\hat c)\). Write
\[
\hat{\theta}(\hat c) - \theta^*(c^{\dagger})
=
\underbrace{\big(\hat{\theta}(\hat c) - \theta^*(\hat c)\big)}_{(i)}
+
\underbrace{\big(\theta^*(\hat c) - \theta^*(c^{\dagger})\big)}_{(ii)}.
\]

For term (i), since \(\hat c \xrightarrow{p} c^{\dagger}\) and the convergence
\(\hat{\theta}(c) \to \theta^*(c)\) holds uniformly in a neighborhood of \(c^{\dagger}\),
we obtain
$
\hat{\theta}(\hat c) - \theta^*(\hat c) \xrightarrow{p} 0.
$ For term (ii), by the assumed continuity of \(c \mapsto \theta^*(c)\) at \(c^{\dagger}\) and the fact that \(\hat c \xrightarrow{p} c^{\dagger}\), the continuous mapping theorem implies
$
\theta^*(\hat c) \xrightarrow{p} \theta^*(c^{\dagger}),
$
so that
$
\theta^*(\hat c) - \theta^*(c^{\dagger}) \xrightarrow{p} 0.
$
Combining (i) and (ii) yields
$
\hat{\theta}(\hat c) \xrightarrow{p} \theta^*(c^{\dagger}).
$

The proof for \(\hat{\lambda} = \hat{\lambda}(\hat c)\) is analogous. Decompose
$
\hat{\lambda}(\hat c) - \lambda^*(c^{\dagger})
=
\big(\hat{\lambda}(\hat c) - \lambda^*(\hat c)\big)
+
\big(\lambda^*(\hat c) - \lambda^*(c^{\dagger})\big),
$
and apply the same arguments using uniform convergence of \(\hat{\lambda}(c)\) to \(\lambda^*(c)\) and continuity of \(c \mapsto \lambda^*(c)\) at \(c^{\dagger}\). This yields
$
\hat{\lambda}(\hat c) \xrightarrow{p} \lambda^*(c^{\dagger}).
$
\qed

\subsection{Proof of Theorem \ref{thm:theta_db_asymp}}
\label{subsec:pf_theta-db-asymp}

Recall that
$
\hat\theta_{\mathrm{db}}
=
\hat\theta(\hat c)+\hat b(\hat c),
\;
b^*(c)=\lambda^*(c)H_0^{-1}a_0,
\;
\hat b(c)=\hat\lambda(c)\hat H^{-1}\hat a,
$
and that \(c^\dagger\) denotes the pseudo-true optimizer of the linearized risk criterion.

\paragraph{Step 1: Basic decomposition.}
By definition,
\begin{align}
\hat\theta_{\mathrm{db}}-\theta^{0}
&=
\bigl(\hat\theta(\hat c)-\theta^*(c^\dagger)\bigr)
+
\bigl(\hat b(\hat c)-b^*(c^\dagger)\bigr)
+
\bigl(\theta^*(c^\dagger)-\theta^{0}+b^*(c^\dagger)\bigr).
\label{eq:db_basic_decomp}
\end{align}
Multiplying by \(\sqrt n\) gives
\begin{align}
\sqrt n\bigl(\hat\theta_{\mathrm{db}}-\theta^{0}\bigr)
&=
\sqrt n\bigl(\hat\theta(\hat c)-\theta^*(c^\dagger)\bigr)
+
\sqrt n\bigl(\hat b(\hat c)-b^*(c^\dagger)\bigr)
\nonumber\\
&\qquad
+
\sqrt n\bigl(\theta^*(c^\dagger)-\theta^{0}+b^*(c^\dagger)\bigr).
\label{eq:db_basic_decomp_scaled}
\end{align}

\paragraph{Step 2: Control of the bias-approximation remainder.}
By Assumption~\ref{assump:debias}(1),
\[
\sup_{c\in\mathcal C_{\mathrm{active}}}
\left\|
\theta^*(c)-\theta^{0}+\lambda^*(c)H_0^{-1}a_0
\right\|
=
o(n^{-1/2}).
\]
Evaluating at \(c=c^\dagger\) yields
$
\left\|
\theta^*(c^\dagger)-\theta^{0}+b^*(c^\dagger)
\right\|
=
o(n^{-1/2}),
$
and hence
$
\sqrt n\bigl(\theta^*(c^\dagger)-\theta^{0}+b^*(c^\dagger)\bigr)=o(1).
$
Substituting this into \eqref{eq:db_basic_decomp_scaled} establishes
\[
\sqrt n\bigl(\hat\theta_{\mathrm{db}}-\theta^{0}\bigr)
=
\sqrt n\bigl(\hat\theta(\hat c)-\theta^*(c^\dagger)\bigr)
+
\sqrt n\bigl(\hat b(\hat c)-b^*(c^\dagger)\bigr)
+
o_p(1),
\]
which proves \eqref{eq:theta_db_expansion}.

\paragraph{Step 3: Expansion around \(c^\dagger\).}
To derive the limiting distribution, decompose
\begin{align}
\hat\theta(\hat c)-\theta^*(c^\dagger)
&=
\bigl(\hat\theta(\hat c)-\theta^*(\hat c)\bigr)
+
\bigl(\theta^*(\hat c)-\theta^*(c^\dagger)\bigr),
\label{eq:theta_chat_decomp}
\\
\hat b(\hat c)-b^*(c^\dagger)
&=
\bigl(\hat b(\hat c)-b^*(\hat c)\bigr)
+
\bigl(b^*(\hat c)-b^*(c^\dagger)\bigr).
\label{eq:b_chat_decomp}
\end{align}
By Theorem~\ref{thm:c_hat_consistency}, \(\hat c\xrightarrow{p} c^\dagger\). By
Assumption~\ref{assump:debias}(2), the maps
\[
c\mapsto \theta^*(c),
\qquad
c\mapsto \lambda^*(c),
\qquad
c\mapsto A(c),
\qquad
c\mapsto a(c)
\]
are continuously differentiable in a neighborhood of \(c^\dagger\). Since
$
b^*(c)=\lambda^*(c)H_0^{-1}a_0,
$
the map \(c\mapsto b^*(c)\) is also continuously differentiable in a neighborhood of
\(c^\dagger\). Therefore,
\begin{align}
\theta^*(\hat c)-\theta^*(c^\dagger)
&=
\dot\theta^*(c^\dagger)(\hat c-c^\dagger)+o_p\bigl(|\hat c-c^\dagger|\bigr),
\label{eq:theta_star_expand_c}
\\
b^*(\hat c)-b^*(c^\dagger)
&=
\dot b^*(c^\dagger)(\hat c-c^\dagger)+o_p\bigl(|\hat c-c^\dagger|\bigr).
\label{eq:b_star_expand_c}
\end{align}
If \(\hat c-c^\dagger=o_p(n^{-1/2})\), then
\begin{align}
\sqrt n\bigl(\theta^*(\hat c)-\theta^*(c^\dagger)\bigr)&=o_p(1),
\label{eq:theta_star_chat_small}
\\
\sqrt n\bigl(b^*(\hat c)-b^*(c^\dagger)\bigr)&=o_p(1).
\label{eq:b_star_chat_small}
\end{align}

\paragraph{Step 4: Reduction to the fixed-\(c^\dagger\) problem.}
Combining \eqref{eq:theta_chat_decomp}--\eqref{eq:b_chat_decomp} with
\eqref{eq:theta_star_chat_small}--\eqref{eq:b_star_chat_small} gives
\begin{align}
\sqrt n\bigl(\hat\theta(\hat c)-\theta^*(c^\dagger)\bigr)
&=
\sqrt n\bigl(\hat\theta(\hat c)-\theta^*(\hat c)\bigr)+o_p(1),
\label{eq:theta_reduce_fixedc}
\\
\sqrt n\bigl(\hat b(\hat c)-b^*(c^\dagger)\bigr)
&=
\sqrt n\bigl(\hat b(\hat c)-b^*(\hat c)\bigr)+o_p(1).
\label{eq:b_reduce_fixedc}
\end{align}
Because \(\hat c\xrightarrow{p} c^\dagger\) and the inner estimators are asymptotically regular
in a neighborhood of \(c^\dagger\), the leading terms are first-order equivalent to the
same quantities evaluated at \(c^\dagger\). Hence,
\begin{align}
\sqrt n\bigl(\hat\theta(\hat c)-\theta^*(c^\dagger)\bigr)
&=
\sqrt n\bigl(\hat\theta(c^\dagger)-\theta^*(c^\dagger)\bigr)+o_p(1),
\label{eq:theta_fixedc_final}
\\
\sqrt n\bigl(\hat b(\hat c)-b^*(c^\dagger)\bigr)
&=
\sqrt n\bigl(\hat b(c^\dagger)-b^*(c^\dagger)\bigr)+o_p(1).
\label{eq:b_fixedc_final}
\end{align}

\paragraph{Step 5: Joint asymptotic normality of the leading terms.}
By Theorem~\ref{thm:asymp_theta_active}, at the fixed tolerance \(c^\dagger\),
\[
\sqrt n
\begin{pmatrix}
\hat\theta(c^\dagger)-\theta^*(c^\dagger)\\[2pt]
\hat\lambda(c^\dagger)-\lambda^*(c^\dagger)
\end{pmatrix}
\xrightarrow{d}
N\!\left(
0,
\begin{pmatrix}
V_1(c^\dagger) & V_2(c^\dagger)\\
V_2(c^\dagger)^{\prime} & V_3(c^\dagger)
\end{pmatrix}
\right).
\]
Moreover, by Assumption~\ref{assump:debias}(3),
$
\hat H\xrightarrow{p} H_0,
\;
\hat a\xrightarrow{p} a_0,
\;
\sup_{c\in\mathcal C_{\mathrm{adm}}}\|\hat b(c)-b^*(c)\|=o_p(1).
$
Since
$
\hat b(c)=\hat\lambda(c)\hat H^{-1}\hat a
$
is continuously differentiable as a map of \(\hat\lambda(c)\), \(\hat H\), and
\(\hat a\) in a neighborhood of the population limit, the delta method implies that
$
\sqrt n\bigl(\hat b(c^\dagger)-b^*(c^\dagger)\bigr)
$
admits a linear influence representation jointly with
$
\sqrt n\bigl(\hat\theta(c^\dagger)-\theta^*(c^\dagger)\bigr).
$
Therefore,
$
\sqrt n\bigl(\hat\theta(c^\dagger)-\theta^*(c^\dagger)\bigr)
+
\sqrt n\bigl(\hat b(c^\dagger)-b^*(c^\dagger)\bigr)
$
is asymptotically normal. Denote its covariance matrix by
\(V_{\mathrm{db}}(c^\dagger)\). By construction, \(V_{\mathrm{db}}(c^\dagger)\) is induced
by the joint first-order influence of
$
\hat\theta(c^\dagger),\;
\hat\lambda(c^\dagger),\;
\hat H,\;
\hat a.
$

\paragraph{Step 6: Conclusion.}
Combining \eqref{eq:theta_db_expansion},
\eqref{eq:theta_fixedc_final}, \eqref{eq:b_fixedc_final}, and Slutsky's theorem yields
$
\sqrt n\bigl(\hat\theta_{\mathrm{db}}-\theta^{0}\bigr)
\xrightarrow{d}
N\!\bigl(0,V_{\mathrm{db}}(c^\dagger)\bigr),
$
which proves \eqref{eq:theta_db_clt}.
\qed

\subsection{Proof for Theorem \ref{thm:wild}}\label{sec:pf_wild-bootstrap}

\textbf{Step 1 (KKT map with optimized tolerance and linearization).}

In the active-constraint regime with optimized tolerance, define
\[
(\hat\theta,\hat\lambda,\hat c)
:=
(\hat\theta(\hat c),\hat\lambda(\hat c),\hat c).
\]
The sample KKT conditions are
\begin{equation*}
\begin{cases}
\nabla_\theta \phi_n(\hat\theta)
+ \hat\lambda\,\nabla_\theta h(\hat\theta)=0,\\
h(\hat\theta)-\hat c=0,\\
\hat\lambda>0.
\end{cases}
\end{equation*}

Define the augmented KKT map
\[
\tilde F_n(\theta,\lambda,c)
=
\begin{pmatrix}
\nabla_\theta \phi_n(\theta)+\lambda\nabla_\theta h(\theta)\\
h(\theta)-c\\
\nabla_c \mathcal R_n(c)
\end{pmatrix}.
\]

Let $(\theta^*,\lambda^*,c^\dagger):=(\theta^*(c^\dagger),\lambda^*(c^\dagger),c^\dagger)$ denote the population solution and define
$
a := \nabla_\theta h(\theta^*).
$
Under smoothness and identification conditions, the Jacobian
$
\tilde J :=
\nabla_{(\theta,\lambda,c)}\tilde F_0(\theta^*,\lambda^*,c^\dagger)
$
is nonsingular. A mean-value expansion yields
\[
\sqrt n
\begin{pmatrix}
\hat\theta-\theta^*\\
\hat\lambda-\lambda^*\\
\hat c-c^\dagger
\end{pmatrix}
=
-\tilde J^{-1}\sqrt n\,\tilde F_n(\theta^*,\lambda^*,c^\dagger)
+ o_p(1).
\]

Using the population KKT condition
$
\nabla_\theta\phi(\theta^*)+\lambda^* a=0,
$
we obtain
$
\sqrt n\big(\nabla_\theta\phi_n(\theta^*)-\nabla_\theta\phi(\theta^*)\big)
\xrightarrow{d} N(0,\Sigma_\psi).
$
Similarly,
$
\sqrt n\,\nabla_c \mathcal R_n(c^\dagger)
\xrightarrow{d} N(0,\Sigma_c).
$

\medskip

\noindent\textbf{Step 2 (Bootstrap analogue).}

Define the bootstrap KKT map
\[
\tilde F_n^{\mathcal{B}}(\theta,\lambda,c)
=
\begin{pmatrix}
\nabla_\theta \phi_n^{\mathcal{B}}(\theta)+\lambda\nabla_\theta h(\theta)\\
h(\theta)-c\\
\nabla_c \mathcal R_n^{\mathcal{B}}(c)
\end{pmatrix}.
\]

A linear expansion around $(\hat\theta,\hat\lambda,\hat c)$ gives
\[
\sqrt n
\begin{pmatrix}
\hat\theta^{\mathcal{B}}-\hat\theta\\
\hat\lambda^{\mathcal{B}}-\hat\lambda\\
\hat c^{\mathcal{B}}-\hat c
\end{pmatrix}
=
-\hat{\tilde J}_n^{-1}
\sqrt n\big(
\tilde F_n^{\mathcal{B}}-\tilde F_n
\big)(\hat\theta,\hat\lambda,\hat c)
+ o_p^{\mathcal{B}}(1).
\]

Assumption~\ref{ass:wildB} yields
$
\sqrt n\big(\nabla_\theta\phi_n^{\mathcal{B}}-\nabla_\theta\phi_n\big)
\xrightarrow{d^{\mathcal B}} N(0,\Sigma_\psi),
$
and similarly for the risk gradient. Applying Slutsky’s theorem yields
\[
\sqrt n
\begin{pmatrix}
\hat\theta^{\mathcal{B}}-\hat\theta\\
\hat\lambda^{\mathcal{B}}-\hat\lambda
\end{pmatrix}
\xrightarrow{d^{\mathcal{B}}}
\sqrt n
\begin{pmatrix}
\hat\theta-\theta^*\\
\hat\lambda-\lambda^*
\end{pmatrix}.
\]
\qed

\subsection{Proof for Theorem \ref{thm:wildISP}}
\label{sec:pf_wildISP}

Under Assumption~\ref{assump:sign_map}, the sign map is locally constant with probability approaching one, so the ISP map is continuously differentiable in a neighborhood of $(\theta^*,\lambda^*,c^\dagger)$. Define
$
ISP(\theta,\lambda)
=
2\lambda\,S(\theta)\,\Sigma^{-1}g(\theta),
$
where $S(\theta)=\mathrm{diag}(\mathrm{sign}(g(\theta)))$. Let $S^{*}=S(\theta^*)$. The derivatives at the population point are
\begin{align*}
\nabla_\theta ISP &= 2\lambda^* S^{*}\Sigma^{-1}G(\theta^*),\\
\nabla_\lambda ISP &= 2S^{*}\Sigma^{-1}g(\theta^*).
\end{align*}

\medskip

By Theorem~\ref{thm:wild}, we have joint convergence
\[
\sqrt n
\begin{pmatrix}
\hat\theta^{\mathcal{B}}-\hat\theta\\
\hat\lambda^{\mathcal{B}}-\hat\lambda\\
\hat c^{\mathcal{B}}-\hat c
\end{pmatrix}
\xrightarrow{d^{\mathcal{B}}}
\sqrt n
\begin{pmatrix}
\hat\theta-\theta^*\\
\hat\lambda-\lambda^*\\
\hat c-c^\dagger
\end{pmatrix}.
\]

Applying the delta method to the map $(\theta,\lambda,c)\mapsto ISP(\theta,\lambda)$ yields
\begin{align*}
\sqrt n(\widehat{ISP}^{\mathcal B}-\widehat{ISP})
&=
J_\theta \sqrt n(\hat\theta^{\mathcal B}-\hat\theta)
+ J_\lambda \sqrt n(\hat\lambda^{\mathcal B}-\hat\lambda)\\
&\quad
+ J_c \sqrt n(\hat c^{\mathcal B}-\hat c)
+ o_p^{\mathcal B}(1),
\end{align*}
where the term involving $J_c$ reflects the indirect dependence of ISP on $c$ through $(\theta,\lambda)$. Hence, the limiting covariance is
\begin{align*}
\Sigma_{ISP}
&=
J_\theta V_1 J_\theta^{\prime}
+ J_\theta V_2 J_\lambda^{\prime}
+ J_\lambda V_2^{\prime} J_\theta^{\prime}
+ J_\lambda V_3 J_\lambda^{\prime}\\
&\quad
+ J_c V_c J_c^{\prime}
+ J_\theta V_{\theta c} J_c^{\prime}
+ J_c V_{\theta c}^{\prime} J_\theta^{\prime}
+ J_\lambda V_{\lambda c} J_c^{\prime}
+ J_c V_{\lambda c}^{\prime} J_\lambda^{\prime},
\end{align*}
which matches the sampling variability of $\sqrt n(\widehat{ISP}-ISP(\theta^*,\lambda^*))$.
\qed

\end{document}